\documentclass[sigconf, nonacm]{acmart}
\usepackage{amsthm}
\usepackage{multirow}
\usepackage{makecell}
\usepackage{soul}
\usepackage{mathtools}
\newtheorem{theorem}{Theorem}
\newcommand{\sstitle}[1]{{\bf #1\/.}}
\theoremstyle{definition}
\newtheorem{mydef}{Definition}
\usepackage{subcaption}
\usepackage[ruled,vlined,linesnumbered]{algorithm2e}

\SetKwFor{Parallel}{parallel}{do}{}

\let\oldnl\nl
\newcommand{\nonl}{\renewcommand{\nl}{\let\nl\oldnl}}
\newcommand{\Break}{\textbf{break}}
\SetKwProg{Fn}{Function}{}{}
\SetKw{Continue}{continue}

\newcommand\vldbdoi{XX.XX/XXX.XX}
\newcommand\vldbpages{XXX-XXX}
\newcommand\vldbvolume{14}
\newcommand\vldbissue{1}
\newcommand\vldbyear{2020}
\newcommand\vldbauthors{\authors}
\newcommand\vldbtitle{\shorttitle} 
\newcommand\vldbavailabilityurl{https://github.com/Spatio-Temporal-Lab/Falcon}
\newcommand\vldbpagestyle{plain} 

\begin{document}
\title{{\em Falcon}: GPU-Based Floating-point Adaptive Lossless Compression}

\author{Zheng Li}
\author{Weiyan Wang}
\affiliation{%
  \institution{Chongqing University, China}
}
\email{zhengli@cqu.edu.cn}
\email{weiyan.wang@stu.cqu.edu.cn}


\author{Ruiyuan Li}
\author{Chao Chen}
\affiliation{%
  \institution{Chongqing University, China}
}
\email{ruiyuan.li@cqu.edu.cn}
\email{cschaochen@cqu.edu.cn}

\author{Xianlei Long}
\author{Linjiang Zheng}
\affiliation{%
  \institution{Chongqing University, China}
}
\email{xianlei.long@cqu.edu.cn}
\email{zlj\_cqu@cqu.edu.cn}

\author{Quanqing Xu}
\author{Chuanhui Yang}
\affiliation{%
  \institution{OceanBase, Ant Group, China}
}
\email{xuquanqing.xqq@oceanbase.com}
\email{rizhao.ych@oceanbase.com}

\begin{abstract}
Domains such as IoT (Internet of Things) and HPC (High Performance Computing) generate a torrential influx of floating-point time-series data. Compressing these data while preserving their absolute fidelity is critical, and leveraging the massive parallelism of modern GPUs offers a path to unprecedented throughput. Nevertheless, designing such a high-performance GPU-based lossless compressor faces three key challenges: 1) heterogeneous data movement bottlenecks, 2) precision-preserving conversion complexity, and 3) anomaly-induced sparsity degradation.
To address these challenges, this paper proposes {\em Falcon}, a GPU-based \textbf{\underline{F}}loating-point \textbf{\underline{A}}daptive \textbf{\underline{L}}ossless \textbf{\underline{CO}}mpressio\textbf{\underline{N}} framework. Specifically, {\em Falcon} first introduces a lightweight asynchronous pipeline, which hides the I/O latency during the data transmission between the CPU and GPU. Then, we propose an accurate and fast float-to-integer transformation method with theoretical guarantees, which eliminates the errors caused by floating-point arithmetic. Moreover, we devise an adaptive sparse bit-plane lossless encoding strategy, which reduces the sparsity caused by outliers. Extensive experiments on 12 diverse datasets show that our compression ratio improves by \textcolor{black}{9.1\%} over the most advanced CPU-based method, with compression throughput \textcolor{black}{$2.43\times$} higher and decompression throughput \textcolor{black}{$2.4\times$} higher than the fastest GPU-based competitors, respectively. 

\end{abstract}

\maketitle

\pagestyle{\vldbpagestyle}
\begingroup\small\noindent\raggedright\textbf{PVLDB Reference Format:}\\
\vldbauthors. \vldbtitle. PVLDB, \vldbvolume(\vldbissue): \vldbpages, \vldbyear.\\
\href{https://doi.org/\vldbdoi}{doi:\vldbdoi}
\endgroup
\begingroup
\renewcommand\thefootnote{}\footnote{\noindent
*Corresponding author: Ruiyuan Li\\
This work is licensed under the Creative Commons BY-NC-ND 4.0 International License. Visit \url{https://creativecommons.org/licenses/by-nc-nd/4.0/} to view a copy of this license. For any use beyond those covered by this license, obtain permission by emailing \href{mailto:info@vldb.org}{info@vldb.org}. Copyright is held by the owner/author(s). Publication rights licensed to the VLDB Endowment. \\
\raggedright Proceedings of the VLDB Endowment, Vol. \vldbvolume, No. \vldbissue\ %
ISSN 2150-8097. \\
\href{https://doi.org/\vldbdoi}{doi:\vldbdoi} \\
}\addtocounter{footnote}{-1}\endgroup

\ifdefempty{\vldbavailabilityurl}{}{
\vspace{.3cm}
\begingroup\small\noindent\raggedright\textbf{PVLDB Artifact Availability:}\\
The source code, data, and/or other artifacts have been made available at \url{\vldbavailabilityurl}.
\endgroup
}
\section{Introduction}\label{sec:introduction}


With the proliferation of sensing devices and high performance computing (HPC), an unprecedented sheer volume of floating-point data are generated, consuming a humongous amount of storage space. For example, modern airliners (e.g., Boeing 787) are equipped with extensive sensor arrays monitoring flight status, where a single flight can generate up to 0.5 terabytes (TB) of raw floating-point sensor data~\cite{jensen2017time}. In the field of HPC, an individual scientific simulation can generate floating-point data volumes reaching terabytes (TB) or even petabytes (PB)~\cite{chen2024fcbench, liu2021high}. 

One of the best ways to reduce the storage costs is to compress these data before we store them~\cite{tang2025improving}. When evaluating the performance of a compression algorithm, key considerations typically include its losslessness, compression ratio, and compression/decompression throughput. \textbf{Lossy compression methods} (e.g., SZ~\cite{di2016fast, liang2018error, zhao2020significantly, liang2022sz3}, Machete~\cite{shi2024machete}, MOST~\cite{yang2023most} and Serf~\cite{li2025serf}) could usually achieve impressive compression ratios but would lose some information, making them inapplicable in some critical scenarios. For example, the integrity of stored data is non-negotiable in databases~\cite{wang2020apache, li2020just}, where any compromise through unauthorized alteration is fundamentally unacceptable to users. In addition, the high-risk characteristic of aviation~\cite{jensen2017time} necessitates absolute data fidelity, as seemingly negligible errors can trigger catastrophic disasters. Although \textbf{general-purpose compression algorithms} (e.g., Gzip~\cite{gzip}, Lz4~\cite{LZ4}) can achieve lossless compression for floating-point data, they yield suboptimal compression ratios due to their inability to exploit inherent patterns in such datasets. Recently, many \textbf{lossless compression algorithms} tailored for floating-point data (e.g., Elf~\cite{li2023elf}, ALP~\cite{afroozeh2023alp} and Camel~\cite{yao2024camel}) achieve superior compression ratios, but their throughput for both compression and decompression remains constrained by CPU (Central Processing Unit)-bound implementations, leaving substantial room for hardware acceleration.


Over recent decades, exponential hardware advancement has positioned GPUs (Graphics Processing Units) as promising accelerators for floating-point compression. Although several \textbf{GPU-accelerated general-purpose lossless compressors}~\cite{nvcomp2023} have emerged, their compression performance remains suboptimal. The approach in~\cite{jamalidinan2025floating} introduces a floating-point transformation framework that first reorganizes records via clustering to co-locate correlated data, followed by GPU-enabled generic compression. While this hybrid scheme moderately improves both compression ratios and throughput, it suffers from underutilized GPU resources during the clustering phase and lacks dedicated encoding/decoding mechanisms for floating-point data. There are also a negligible number of \textbf{GPU-based floating-point-specified compression methods}~\cite{9910073, o2011floating, yang2015mpc}. However, they are specifically targeted at scientific floating-point data and lack deliberate CPU-GPU co-design, resulting in suboptimal compression ratios and compromised throughput for generic floating-point datasets.

\begin{figure}
  \centering
  \includegraphics[width=\linewidth]{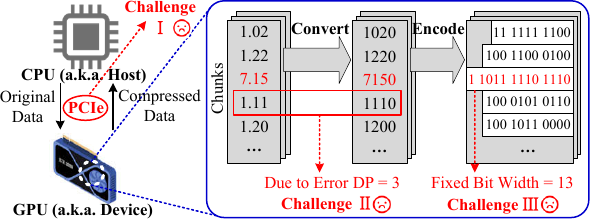}
  \caption{Challenges of GPU-Based Lossless Compression.}
  \label{fig:challenge}
  \vspace{-10pt}
\end{figure}



GPUs, though possessing powerful parallelism, exhibit inferior logic operation capabilities compared to CPUs. Existing CPU-based floating-point compressors are incompatible with GPU architectures due to their fundamentally different execution models. As illustrated in Fig.~\ref{fig:challenge}, designing a GPU-based lossless compression algorithm for floating-point data faces three key challenges.  

\textbf{Challenge~\uppercase\expandafter{\romannumeral1}, Heterogeneous Data Movement Bottlenecks}. From the perspective of the entire compression workflow, original floating-point data are read from the external storage by the CPU (a.k.a. \textbf{Host}) and then transferred to the GPU (a.k.a. \textbf{Device}) via the PCIe (Peripheral Component Interconnect Express) (a process called \textbf{H2D}). After the GPU finishes the compression tasks, the compressed data are returned to the CPU from the GPU via the same PCIe (a process called \textbf{D2H}), and finally written to the external storage by the CPU. In this process, PCIe data transmission usually acts as a bottleneck, during which the CPU or GPU may enter an idle state, leading to a reduction in overall system throughput. Therefore, the first challenge is how to fully utilize the PCIe bandwidth through CPU-to-GPU coordination to maximize the overall throughput.

\textbf{Challenge~\uppercase\expandafter{\romannumeral2}, Precision-Preserving Conversion Complexity}. Floating-point values have complex underlying layouts according to the IEEE 754 standard~\cite{kahan1996ieee}. Most existing floating-point compression methods incorporates complicated logic operations to achieve good compression ratio. To avoid GPU-unfriendly operations, one strategy converts floating-point values to integers followed by integer compression. This necessitates decimal place (cf. Definition~\ref{def:DP}) tracking to ensure lossless reconstruction. For example, the decimal place of 1.02 is 2. During compression, 1.02 is converted to 102 since $1.02 \times 10^2 = 102$. During decompression, we recover 1.02 by $102 \div 10^2$. However, it is non-trivial for machines to calculate the decimal places exactly and efficiently. One time-consuming method is to first convert 1.02 to a string representation ``1.02", followed by the analysis of the resulting string. To accelerate the calculation, the works~\cite{kuschewski2023btrblocks, Elf+} adopt a trial-and-error method, which gradually increases $i$ from 0 and checks whether $1.02 \times 10^i = \lfloor 1.02 \times 10^i\rfloor$. Once $i = 2$, the equation satisfies, so its decimal place is considered to be 2. Although fast, this method might introduce some errors due to the intrinsic imprecision of floating-point arithmetic. For example, most programming languages consider $1.11 \times 10^2$ to be $ 111.00000000000001$ and $1.11 \times 10^3$ to be $1110.0$, so they incorrectly assume that the decimal place of 1.11 is 3 ({\bf for clarify, in the following of this paper, we use $\otimes$ to denote the multiplication operation that may introduce errors, while reserving the term $\times$ for the multiplication in mathematics}). To minimize the storage overhead of decimal place tracking, we can only store the maximum decimal place per data chunk, which introduces systemic vulnerability: any single-value computation error propagates catastrophically throughout the entire chunk, as shown in Fig.~\ref{fig:challenge}. Hence, computing the decimal places exactly and efficiently plays a vital role in enhancing the compression performance.

\textbf{Challenge~\uppercase\expandafter{\romannumeral3}, Anomaly-Induced Sparsity Degradation}. Floa-ting-point datasets frequently contain outliers that would be exaggerated during integer conversion. When converted to integer bit planes, these outliers induce severe sparsity that catastrophically degrades compression performance. As shown in Fig.~\ref{fig:challenge}, although 7.15 constitutes a minor floating-point anomaly, its integer conversion 7150 creates a prominent outlier. To minimize bit-width specification overhead and leverage GPU parallelism, bit-aligned plane representation can be employed, i.e., we represent each integer with fixed 13 bits, causing excessive leading-zero bits for all values except 7150. Therefore, how to handle these outlier records to further improve compression performance is also a major challenge.

In this paper, we propose {\em Falcon}, a GPU-based \textbf{\underline{F}}loating-point \textbf{\underline{A}}daptive  \textbf{\underline{L}}ossless \textbf{\underline{CO}}mpressio\textbf{\underline{N}} framework. {\em Falcon} is composed of a series of meticulously designed stages that are highly parallelizable and tailored for the GPU architecture. Specifically, to address \textbf{Challenge~\uppercase\expandafter{\romannumeral1}}, {\em Falcon} introduces a lightweight asynchronous pipeline to effectively hide the I/O latency between the CPU and GPU. To address \textbf{Challenge~\uppercase\expandafter{\romannumeral2}}, {\em Falcon} proposes an accurate and fast float-to-integer transformation scheme with theoretical guarantees to eliminate the errors introduced by floating-point arithmetic. To address \textbf{Challenge~\uppercase\expandafter{\romannumeral3}}, {\em Falcon} constructs an adaptive sparse bit-plane encoding strategy to reduce the sparsity cased by outliers.
Overall, the main contributions of this paper are as follows:

(1)~We propose a novel asynchronous compression pipeline with a lightweight event-driven scheduler, which optimizes the utilization of CPU and GPU resources during PCIe data transmission and improves the compression throughout remarkably.
    
(2)~We propose a fast and accurate method with theoretical guarantees for computing the decimal places of floating-point data, eliminating the errors caused by floating-point arithmetic. Based on this, we further devise a lossless conversion scheme from floating-point to integer representation, which improves both compression ratio and compression throughout tremendously.
    
(3)~We design an adaptive sparse bit-plane lossless encoding strategy, which reduces the sparsity caused by outliers and enhances the compression ratio greatly.

(4)~Extensive experiments on \textcolor{black}{12} diverse floating-point datasets show that, compared to \textcolor{black}{10} advanced competitors, {\em Falcon} achieves the best compression ratios (average \textcolor{black}{0.299}) and leads the maximum throughput for both compression and decompression (average \textcolor{black}{10.82} GB/s and 12.32 GB/s, respectively). Notably, compared to the best competitors, {\em Falcon} enjoys an relative improvement of \textcolor{black}{9\%} in compression ratio, \textcolor{black}{$2.43\times$} in compression throughout, and \textcolor{black}{$2.4\times$} in decompression throughout, respectively.

The remainder of this paper is organized as follows. In Section~\ref{sec:preliminaries}, we present some preliminaries. We respectively introduce {\em Faclon} compression and decompression in Section~\ref{sec:compression} and Section~\ref{sec:decompression}. The experimental results are shown in Section~\ref{sec:experiment}. In Section~\ref{sec:relatedwork}, we review the related works. Finally, we conclude this paper in Section~\ref{sec:conclusion}.

\vspace{-10pt}
\section{Preliminaries}\label{sec:preliminaries}

This section introduces the floating-point layout and some definitions, and provides the knowledge about the GPU execution model. Table~\ref{tbl:symbols} lists the symbols used frequently throughout this paper.

\setlength{\tabcolsep}{0.15em} 
\begin{table}[t]
\small
\caption{Symbols and Their Meanings}
\vspace{-10pt}
\centering\label{tbl:symbols}
\resizebox{3.2in}{!}{
\begin{tabular}{|c|l|} 
\hline
\textbf{Symbols}			&\textbf{Meanings}	\\
\hline
\hline
$v, z$  & Floating-point value, converted integer\\
\hline
$s$, $\vec{\boldsymbol{e}} = \langle e_1e_2...e_{11}\rangle$, & Sign bit, exponent bits, mantissa bits under IEEE\\
$\vec{\boldsymbol{m}} = \langle m_1m_2...m_{52}\rangle$     & 754 format, where $s, e_i, m_j \in \{0, 1\}$\\
\hline
$\alpha = DP(v), \beta = DS(v)$ & Decimal place and decimal significand of $v$\\
\hline
$v^{(i)} = v \times 10^{i}$ & Decimally scaled value in mathematics\\
\hline
$V = \{ v_{1}, v_{2}, ..., v_{n} \}$ & A set of $n$ floating-point values\\
\hline
$Z = \{ z_1, z_2, ..., z_n \}$ & A set of $n$ converted integer values\\
\hline
$ULP(v)$ & Unit in the Last Place of $v$\\
\hline
\end{tabular}
}
\vspace{-10pt}
\end{table}

\begin{figure}
  \centering
  \includegraphics[width=3.5in]{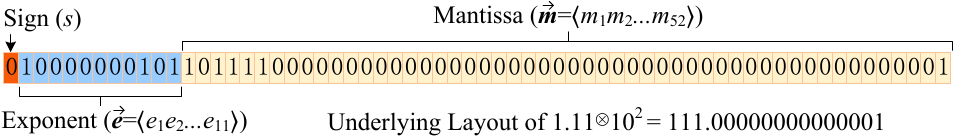}  
  \vspace{-20pt}
  \caption{Floating-Point Layout of Double Values.}
  \label{fig:ieee754}  
  \vspace{-10pt}
\end{figure}

\subsection{Floating-Point Layout}\label{subsec:fplayout}

A floating-point value $v$ can be double-precision (i.e., double value), single-precision (i.e., single value) or half-precision, etc. This paper mainly focuses on \textbf{double values} as they are the most widely applied, but {\em Falcon} can be easily extended to other precisions since all precisions share a similar underlying layout. As shown in Fig.~\ref{fig:ieee754}, according to IEEE 754 standard~\cite{kahan1996ieee}, a double value is stored with 64 binary bits, where 1 bit is for the sign $s$ (i.e., `0' stands for positive and `1' represents negative), 11 bits for the exponent $\vec{\boldsymbol{e}} = \langle e_1e_2... e_{11}\rangle$ (controlling the position of the decimal point), and 52 bits for the mantissa $\vec{\boldsymbol{m}} = \langle m_1m_2...m_{52}\rangle$ (recording the decimal significand of the value shown in Definition~\ref{def:DP}). Specifically, if all 64 bits are zeros, $v = 0$. Otherwise, we have~\footnote{In fact this is only true for normal numbers. Special numbers~\cite{li2023elf, kahan1996ieee} (which are rarely used) are out of the discussion scope of this paper.}:
\begin{equation}\label{equ:floating}
v = (-1)^{s} \times 2^{e - 1023} \times (1 + \sum_{i=1}^{52} m_i \times 2^{-i})
\end{equation}

\noindent where $e = \sum_{i=1}^{11} e_i \times 2^{11-i}$ (we also have $e = \lfloor log_2|v| \rfloor + 1023$). The 1 in the rightmost parenthesis in Equ.~(\ref{equ:floating}) is a hidden bit, which can be regarded as $m_0$ (i.e., there are in fact 53 bits for the decimal significand). For each $m_i$ where $1 \leq i \leq 52$, it contributes $2^{e-1023} \times m_i \times 2^{-i} = 2^{\lfloor log_2|v| \rfloor} \times m_i \times 2^{-i}$ to the value of $v$.

\sstitle{Discussion} Due to the limited bits, Equ.~(\ref{equ:floating}) actually contains certain approximation, which might introduce some errors for floating-point arithmetic. The introduced error would be reflected in the last mantissa bit (i.e., $m_{52}$). For example, $v = 1.11 \times 10^2 = 111.0$ mathematically (for $111.0$, its $m_{52} = 0$), but in computers $v = 1.11 \otimes 10^2 = 111.00000000000001$, which sets $m_{52}$ to 1 as shown in Fig.~\ref{fig:ieee754}. If $v \neq 0$, the introduced error will be no more than the value that $m_{52}$ can contribute, a value also called Unit in the Last Place (ULP) of $v$, denoted by $ULP(v) = 2^{\lfloor log_2|v| \rfloor} \times 2^{-52}$.
For any normal number $v$, we have $|v| \times 2^{-52} = m \times 2^{e - 1023} \times 2^{-52} = m \times ULP(|v|)$, where $m = 1 + \sum_{i=1}^{52} m_i \times 2^{-i}$ and $1 \leq m < 2$. Hence, $ULP(|v|) \leq |v| \times 2^{-52} < 2\times ULP(|v|) \Leftrightarrow ULP(|v|) \times 2^{52} \leq |v| < 2\times ULP(|v|) \times 2^{52}$. 

\subsection{Definitions}
\begin{mydef}
\textbf{Decimal Format~\footnote{For clarity, the decimal format, decimal place and decimal significand used in this paper differ from these in \cite{li2023elf}, respectively.}.} Given a floating-point value $v$, its decimal format is $DF(v) = \pm(d_{h-1}d_{h-2}...d_0.d_{-1}d_{-2}...d_l)_{10}$, where for $l \leq i \leq h-1$, $d_i \in \{$`$0$', `$1$', ..., `$9$'$\}$, $d_{h-1} \neq $`$0$' unless $h = 1$, and $d_l \neq $`$0$'. That is, $DF(v)$ would not start with `0' except that $h = 1$, and would not end with `0'. `$\pm$' means `$+$' or `$-$' denoting the sign of $v$. If $v \geq 0$, `$+$' is usually omitted.
\end{mydef}

For example, $DF(1) = (1.)_{10}$ and $DF(-0.20) = -(0.2)_{10}$. 

\begin{mydef}\label{def:DP}
\textbf{Decimal Place and Decimal Significand.} Given a floating-point value $v$ with its decimal format $DF(v) = \pm(d_{h-1}d_{h-2}$ $...d_0.d_{-1}d_{-2}...d_l)_{10}$, we define $DP(v) = |l|$ as its decimal place. If $v \neq 0$ and for all $ h - 1 \geq i \geq k > l$, $d_i =$ `$0$' but $d_{k-1} \neq 0$ (i.e., $d_{k-1}$ is the first digit from the left that is not equal to `0'), we define $DS(v) = k - l$ as its decimal significand. If $v = 0$, we let $DS(v) = 0$.
\end{mydef}

For example, $DP(-0.0314) = 4$, $DS(-0.0314) = 3$; $DP(1.11) = 2$, $DS(1.11) = 3$; $DP(111.) = 0$, $DS(111.) = 3$; $DP(111.00000000000001) \\= 14$, $DS(111.00000000000001) = 17$.

Given a floating-point value $v \neq 0$, its decimal place and decimal significand have the following relation:
\begin{equation}\label{equ:relation}
DS(v) = DP(v) + \lfloor log_{10}|v| \rfloor + 1
\end{equation}

We use $\alpha$ and $\beta$ to respectively represent the decimal place and decimal significand of a floating-point value if we consider them to be variables or we do not care about the floating-point value $v$.

%

\begin{figure}[t]
  \centering
  \includegraphics[width=3.2in]{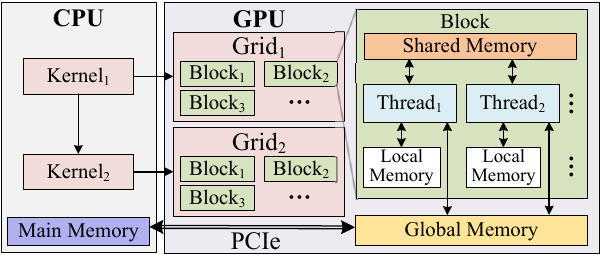}
  \vspace{-5pt}
  \caption{Illustration of GPU Execution Model.}
  \label{fig:gpu_arch}
  \vspace{-10pt}
\end{figure}

\begin{figure*}
  \centering
  \includegraphics[width=7in]{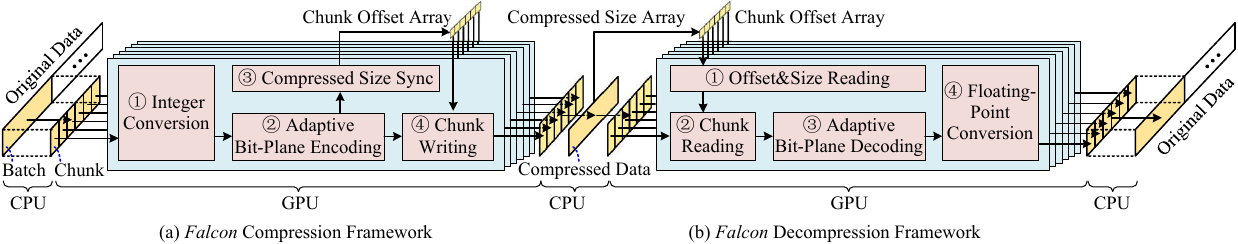}
  \vspace{-20pt}
  \caption{Framework of {\em Falcon}.}
  \vspace{-10pt}
  \label{fig:framework}
\end{figure*}

\subsection{GPU Execution Model}
{\em Falcon} is designed to exploit the architecture of modern GPUs. Understanding the GPU execution model is key to appreciating our design choices. Fig.~\ref{fig:gpu_arch} illustrates the high-level framework.

\sstitle{Execution Hierarchy} A GPU program (a.k.a. \textit{Kernel}) is launched by the CPU and executed on the GPU by a grid of {\em blocks}. Each block contains a group of {\em threads}. Threads are scheduled in a warp (typically 32 threads), which executes in a Single Instruction, Multiple Threads (SIMT) fashion. The massive parallelism is ideal for data-parallel tasks like compression.

\sstitle{Memory Hierarchy} The GPU features a memory hierarchy. \textit{Global memory} is large but has high latency. \textit{Shared memory} is small, low-latency, and shared among threads within a block, making it perfect for inter-thread collaboration. \textit{Local memory} is private to each thread. Efficient GPU algorithms must carefully manage data movement across this hierarchy to minimize latency.

\sstitle{Host-Device Interaction} The CPU (a.k.a. {\em Host}) and the GPU (a.k.a., {\em Device}) communicate over the PCIe. This communication is usually a significant performance bottleneck. 

\section{{\em Falcon} Compression}\label{sec:compression}
In this section, we introduce {\em Falcon} compression. As shown in Fig.~\ref{fig:framework}(a), {\em Falcon} compression is an {\bf asynchronous pipeline}, where the CPU first sequentially reads a batch of original data (called a \textbf{batch}) from external storage, and then copies this batch from the main memory to the global memory via PCIe (a.k.a. H2D). In the global memory, the batch is logically partitioned into several equal-sized disjoint {\bf chunks}, each of which is processed by one GPU thread in four steps: 1)~\textbf{Integer Conversion}, which converts the floating-point values in the corresponding chunk into integers; 2)~\textbf{Adaptive Bit-Plane Encoding}, which encodes the converted integers into compressed binary bits; 3)~\textbf{Compressed Size Sync}, which synchronizes the compressed chunk size among threads, so we can get the chunk offset array; and 4)~\textbf{Chunk Writing}, which directly writes the compressed data into the right position in the global memory with the help of the chunk offset array. After the four steps are finished, the CPU copies the compressed data and the chunk offset array into the main memory via PCIe (a.k.a. D2H), and finally writes them into the external storage. Next, we first introduce the pipeline, and then present each step in detail.

\subsection{Asynchronous Pipeline}\label{subsec:pipeline}


The primary bottleneck in GPU-based compression stems from the I/O latency between the CPU and GPU. To this end, {\em Falcon} implements an asynchronous pipeline to hide this latency, where the pipeline is executed at the granularity of data batches. However, the compressed size of each batch is generally unpredictable prior to execution. When multiple batches are processed concurrently, their results may be written to memory simultaneously, potentially leading to conflicts in memory access.

\begin{figure}[t]
  \centering
  \includegraphics[width=3.4in]{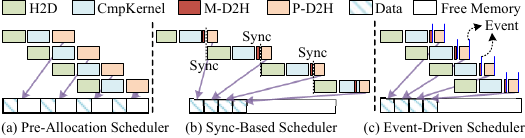}
  \vspace{-15pt}
  \caption{Different Schedulers of Pipeline.}
  \label{fig:DifferentSchedulers}
  \vspace{-10pt}
\end{figure}

To address this problem, the most straightforward approach is Pre-Allocation Scheduler, as shown in Fig.~\ref{fig:DifferentSchedulers}(a), which pre-allocates a fixed-size memory space for the compressed results of each batch. Upon completion of processing, the results are copied directly into the corresponding pre-allocated space. However, if the allocated space is too large, it leads to memory wastage; if too small, it risks memory conflicts. Besides, since the CPU is unaware of the actual compressed size, it may copy a substantial amount of useless data during D2H, thereby wasting valuable PCIe bandwidth. Additionally, this approach requires an extra merging step, which reduces throughput. An alternative solution is Sync-Based Scheduler, as shown in Fig.~\ref{fig:DifferentSchedulers}(b), which decouples D2H into two phases: M-D2H (synchronizing the compressed size) and P-D2H (synchronizing the actual compressed data). After the GPU compression (i.e., CmpKernel) of the previous batch completes, the compressed size is first sent to the CPU via M-D2H. The CPU can immediately launches the compression of the next batch, since the offset for storing the next batch’s results can now be determined. This allows the P-D2H of the previous batch to overlap with the H2D of the next batch, leveraging the bidirectional communication capability of PCIe. However, this method still limits throughput, as the next batch can only start after the M-D2H step of the previous batch finishes.

\begin{figure}[t]
  \centering
  \includegraphics[width=3.4in]{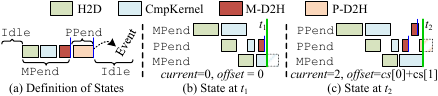}
  \vspace{-15pt}
  \caption{Illustration of Event-Driven Scheduler.}
  \label{fig:mutiStreamState}
  \vspace{-10pt}
\end{figure}

To resolve the issues present in the two aforementioned methods, this paper proposes Event-Driven Scheduler that makes full use of GPU event mechanisms, as shown in Fig.~\ref{fig:DifferentSchedulers}(c). Specifically, Event-Driven Scheduler divides the execution of a batch into three distinct states, as depicted in Fig.~\ref{fig:mutiStreamState}(a). Initially, the batch remains in {\tt Idle} before the processing begins. Once the batch starts and proceeds through H2D, CmpKernel and M-D2H, its state transitions to {\tt MPend}, indicating that the CPU is waiting for its meta compressed data size. Upon completion of M-D2H, the GPU immediately sends the first event to the CPU, prompting the batch to enter the {\tt PPend} state, meaning that the CPU is waiting for the compressed data. Finally, after P-D2H completes, the GPU sends the second event to the CPU, signaling the completion of the batch and returning it to the {\tt Idle} state. The essential is to determine the compressed data offsets of batches based on their launch order.

\begin{algorithm}[t]
\small
	\caption{$EventDrivenScheduler(in, N_{s}, batchSize)$}\label{alg:EventDrivenScheduler}
	\tcc{Initialization}
	$stream[N_{s}], state[N_{s}] \leftarrow \{{\tt Idle}\}, *gIn[N_{s}], *gOut[N_{s}]$\;
	$eventQ[N_{s}]$, $cs[N_s]$, $current \leftarrow 0$, $active \leftarrow 0$\;
	$*cache, offset \leftarrow 0$, $batch \leftarrow in.read(batchSize)$\;
	\tcc{Verification Loop}
	\While{$batch \neq null$ {\rm or} $active > 0$}{
		\For{$i$ {\rm from }$0$ {\rm to }$(N_{s}-1)$}{
			\If{$state[i] = {\tt Idle}$}{
				$H2DAsync(batch, gIn[i], stream[i])$\;
				$CmpKernelAsync(gIn[i], gOut[i], stream[i])$\;
				$MD2HAsync(gOut[i], cs[i], stream[i])$\;
				$SendEventAsync(eventQ[i], stream[i])$\;
				$active \leftarrow active + 1$; $state[i] \leftarrow {\tt MPend}$\;
			}
			\ElseIf{$state[i] = {\tt MPend}$}{
				\If{$current = i$ {\tt and} $ReceiveEvent(eventQ[i])$}{
					$PD2HAsync(gOut[i], cache + offset, stream[i])$\;
					$SendEventAsync(eventQ[i], stream[i])$\;
					$current \leftarrow (current + 1)\% N_s$\; 
					$offset \leftarrow offset + cs[i]$;
					$state[i] \leftarrow {\tt PPend}$\;
				}
			}
			\ElseIf(\tcp*[h]{{\tt PPend}} state){$ReceiveEvent(eventQ[i])$}{
				$batch \leftarrow in.read(batchSize)$\;
				$active \leftarrow active - 1$; $state[i] \leftarrow {\tt Idle}$; 
			}
		}
	}
	\Return{$(cache, offset)$}\;
\end{algorithm}

Algorithm~\ref{alg:EventDrivenScheduler} presents the pseudocode of Event-Driven Scheduler, which takes as input the stream of original values  $in$, the number of pipeline streams $N_s$, and the number of values per batch $batchSize$. It returns the starting address of the compressed data and the total size after compression. The algorithm is broadly divided into two parts: {\tt Initialization} and {\tt Verification Loop}.

In {\tt Initialization}, we declare several key variables. Specifically, $stream[N_s]$ and $state[N_s]$ represent the $N_s$ streams and their corresponding states (initialized to {\tt Idle}), respectively. $gIn[N_s]$ and $gOut[N_s]$ are allocated in the global memory of GPU, storing the input (pre-compression) and output (post-compression) data for each stream. Additionally, $eventQ[N_s]$ denotes the event queue associated with each stream, $cs[N_s]$ stores the compressed size for each stream, $current$ indicates the index of the stream currently awaiting the compressed size (referred to as the current stream), $active$ records the number of active streams (i.e., not {\tt Idle}), $cache$ represents the starting address of the compressed results, $offset$ tracks the offset after obtaining the compressed result of the current stream, and $batch$ holds the raw data for the upcoming batch.

In {\tt Verification Loop}, the algorithm continuously iterates through each stream as long as there remains unprocessed data ($batch \neq null$) or active streams ($active > 0$). For each stream examined, corresponding operations are performed based on its state. Due to space limitations and the self-explanatory nature of the pseudocode, a detailed step-by-step explanation is omitted here. Note that all functions labeled with ``Async'' are asynchronous; the CPU does not wait for their completion before proceeding to the next line of code. The last parameter of each asynchronous function specifies the stream to which it belongs, and all asynchronous operations within the same stream are executed sequentially.

Fig.~\ref{fig:mutiStreamState}(b) gives an example of pipeline with three streams. At time $t_1$, although Stream 1 has completed M-D2H stage, it must wait for Stream 0 to finish its M-D2H, because the CPU cannot yet determine its result offset. At time $t_2$, as shown in Fig.~\ref{fig:mutiStreamState}(c), both Stream 0 and Stream 1 have completed their M-D2H stages, so $current$ is updated to 2, and $offset$ is set to $cs[0] + cs[1]$. The P-D2H stage of Stream 1 may be executed before that of Stream 0.

\subsection{Integer Conversion}\label{subsec:integerConversion}


In accordance with IEEE 754 standard~\cite{kahan1996ieee}, floating-point values have rather complex underlying layouts, where a small fluctuation in the value will result in a significant change at the bit level. To enhance the compression ratio, existing CPU-based floating-point compression methods adopt many complicated logic operations that are not friendly to the GPU parallelization. This is because the GPU works in a Single Instruction, Multiple Threads (SIMT) fashion. If there are too many logical branches, it may cause the same threads within a warp to execute sequentially, severely degrading throughput. To this end, this paper devises a GPU-friendly floating-point compression method. Intuitively, given a floating-point value $v$ with its decimal place $\alpha = DP(v)$, we can convert it to a more compressible integer\footnote{The converted integer from a double value also occupies 64 bits.} $round(v \times 10^\alpha)$, and then encode the integer with a GPU-friendly encoding strategy. However, during the conversion process, there are still two problems. {\bf Problem~\uppercase\expandafter{\romannumeral1}}: Can the original value be recovered losslessly by the converted integer? {\bf Problem~\uppercase\expandafter{\romannumeral2}}: How to calculate the decimal place of a floating-point value exactly and efficiently? 

\subsubsection{Correctness of Conversion}\label{subsubsec:correctConversion}

In this subsection, we assume that the decimal place $\alpha$ of $v$ is already given. For the convenience of expression, we provide the definition of {\bf Decimally Scaled Value}.

\begin{mydef}\label{def:dsv}
\textbf{Decimally Scaled Value.} Given a floating-point value $v$, $\alpha = DP(v)$, if we only consider the mathematical rule, we call $v^{(i)} = v \times 10^{i}$ the decimally scaled value of $v$, where $i \geq 0$.
\end{mydef}

\begin{theorem}[Decimal Significand Invariance]\label{the:keepbeta}
Given $v$ with $\alpha = DP(v)$, for any $i \in \{0, 1, ..., \alpha\}$, we have $DS(v) = DS(v^{(i)})$.
\end{theorem}
\begin{proof}
From the mathematical view, for any $0 \leq i \leq \alpha$, $v^{(i)}$ merely shifts the decimal point of $v$ to the right by $i$ positions, i.e., it does not change the value of $k$ and $l$ in Definition~\ref{def:DP}. Therefore, from the mathematical view, $DS(v) = DS(v^{(i)})$.
\end{proof}

Although from the mathematical view, Theorem~\ref{the:keepbeta} keeps correct, due to the limited bits to represent a floating-point value, $v \otimes 10^i$ may introduce some error, i.e., it would round the overflowing bits (i.e., $m_{53}m_{54}...$ in Fig.~\ref{fig:ieee754} that we can image) to $m_{52}$. We call this {\bf binary round}. To get the mathematically correct converted integer $v^{(\alpha)}$, we perform another round operation (called {\bf decimal round}) on $v \otimes 10^\alpha$, i.e., $round(v \otimes 10^\alpha)$. We now prove that $round(v \otimes 10^\alpha) = v^{(\alpha)}$ under a certain condition.

\begin{theorem}[Conversion Correctness]\label{the:round}
Given $v$ with $\alpha = DP(v)$ and $\beta = DS(v)$, if $\beta \leq 15$, then $round(v \otimes 10^\alpha) = v^{(\alpha)}$.
\end{theorem}
\begin{proof}
According to Theorem~\ref{the:keepbeta} (Decimal Significand Invariance), $DS(v^{(\alpha)}) = DS(v) = \beta \leq 15$, so we let $DF(v^{(\alpha)}) = \pm(d_{\beta-1}d_{\beta-2}\\...d_0.)_{10}$. 
If $v = 0$, $round (0 \otimes 10) = 0 = v^{(\alpha)}$.
If $v \neq 0$, we have: $1 \leq |v|^{(\alpha)} < 10^\beta \leq 10^{15} \Rightarrow log_2{1} \leq log_2{|v|^{(\alpha)}} < log_2{10^{15}} \Rightarrow 0 \leq \lceil log_2(|v|^{(\alpha)})\rceil \leq \lceil log_2{10^{15}} \rceil = 50$. That is, $v^{(\alpha)}$ can be represented by no more than 50 binary bits. Apart from the hidden bit, it can be represented by $m_1 \thicksim m_{49}$, i.e., the $m_{50} \thicksim m_{52}$ of its underlying layout are all supposed to be `0's. As shown in Fig.~\ref{fig:round}, suppose $m_t$ is the last bit that equals `1' in the underlying storage of $v^{(\alpha)}$, and $m_k$ is the first bit of the fractional part in $v^{(\alpha)}$, then $t < k \leq 50$. Due to the imprecision of floating-point arithmetic, there are three cases for the value of $m_{t}m_{t+1}...m_k...m_{51}m_{52}$ in $v \otimes 10^\alpha$, as shown in the right of Fig.~\ref{fig:round}. 
\textbf{Case~1: No Error}. In this case, $v \otimes 10^{\alpha} = v^{(\alpha)}$ (e.g., $v = 2.5$), and the decimal round operation imposes no effect on $v \otimes 10^\alpha$. Hence, $round(v \otimes 10^\alpha) = v^{(\alpha)}$. 
\textbf{Case~2: Error Up}. In this case, $|v| \otimes 10^{\alpha} > |v|^{(\alpha)}$ (e.g., $v = 1.11$), and the decimal round operation leaves out the bits after $m_{k-1}$ directly. Therefore, $round(v \otimes 10^\alpha) = v^{(\alpha)}$. 
\textbf{Case~3: Error Down}. In this case, $|v| \otimes 10^{\alpha} < |v|^{(\alpha)}$ (e.g., $v = 8.04$), and the decimal round operation not only removes the bits after $m_{k-1}$, but also increments $m_{k-1}$ by 1, finally propagating to $m_t$ and obtaining $v^{(\alpha)}$, i.e., $round(v \otimes 10^\alpha) = v^{(\alpha)}$.
\end{proof}

\begin{figure}[t]
  \centering
  \includegraphics[width=3.45in]{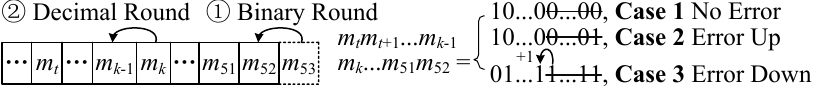}
  \vspace{-10pt}
  \caption{Illustration of Round.}
  \label{fig:round}
  \vspace{-10pt}
\end{figure}

If $\beta > 15$, Theorem~\ref{the:round} may not hold, since the binary round might affect the decimal round. The value 9.110900773177071 constitutes a counterexample, as in computers $round(9.110900773177071 \otimes 10^{15}) = 9110900773177072 \neq 9.110900773177071^{(15)}$.

\begin{theorem}[Conversion Recoverability]\label{the:recoverable}
Given $v$ with $\alpha = DP(v) \leq 22$ and $\beta = DS(v) \leq 15$, we have $v = round(v \otimes 10^\alpha) \div 10^\alpha$, i.e., the floating-point arithmetic would not introduce any error\footnote{The floating-point arithmetic error is defined as: the underlying storage of the mathematical result differs from that of the floating-point arithmetic result.}.
\end{theorem}
\begin{proof}
According to IEEE 754 standard~\cite{kahan1996ieee}, whether a floating-point arithmetic introduce errors depends on whether the involved values can be fully represented by the floating-point storage (i.e., without involving binary round). 
{\bf Firstly}, based on Theorem~\ref{the:round} (Conversion Correctness), if $\beta \leq 15$, $round(v \otimes 10^\alpha) = v^{(\alpha)}$ is a representable value without any binary round error using the floating-point underlying layout. 
{\bf Secondly}, for $10^\alpha$, since $10^\alpha = 2^\alpha \times 5^\alpha$ and $2^\alpha$ can be represented by the exponent part of a floating-point value, hence we only investigate how many mantissa bits are required to represent $5^\alpha$. To let the mantissa bits store $5^\alpha$ without any binary round error, it must satisfy: $\lceil log_2{5^\alpha} \rceil \leq 53 \Leftrightarrow log_2{5^\alpha} \leq 53 \Leftrightarrow \alpha \leq 53 \div log_2{5} \approx 22.83$. That is, if $\alpha \leq 22$, $10^\alpha$ can be represented by the 52 mantissa bits and the 1 hidden bit without any binary round error.
\end{proof}

If $\beta \leq 15$ and $\alpha > 22$, Theorem~\ref{the:recoverable} may not hold. For example, suppose $v = 1.23456789876543 \times 10^{-9}$, then $\alpha = DP(v) = 23$, $\beta = DS(v) = 15$, but $round(1.23456789876543 \times 10^{-9} \otimes 10^{23}) \div 10^{23} = 1.2345678987654302 \times 10^{-9} \neq v$ in computers. Users may question why $1.11 \otimes 10^2$ will introduce error. This is because 1.11 cannot be fully represented by the floating-point storage, which violates the condition that all involved values are without binary round.

\sstitle{Summary} Considering both Theorem~\ref{the:round} and Theorem~\ref{the:recoverable}, given $v$ with $\beta = DS(v)$ and $\alpha = DP(v)$, \ul{if $\beta \leq 15$ and $\alpha \leq 22$, we convert $v$ to $round(v \otimes 10^\alpha)$ during compression, and recover $v = round(v \otimes$ $10^\alpha) \div 10^\alpha$ during decompression}. In Section~\ref{subsubsec:digittransform}, we will propose another conversion method for the case of $\beta > 15$ or $\alpha > 22$.

\subsubsection{Decimal Place Calculation}\label{subsubsec:dpcal}
As discussed in Section~\ref{sec:introduction}, it is challenging to calculate the decimal place of a floating-point value. Existing methods may encounter the efficiency or accuracy problem. To this end, in this paper, we propose a novel method to quickly and accurately calculate the decimal places of floating-point values. The proposed method is based on the following theorem.

\begin{figure*}
  \centering
  \makebox[\textwidth][c]{\includegraphics[width=7.5in]{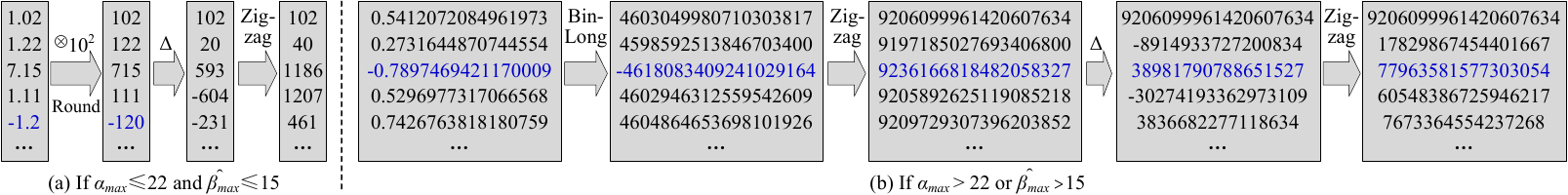}}
  \vspace{-15pt}
  \caption{Illustration of Digit Transformation.}
  \label{fig:IntegerConversion}
  \vspace{-10pt}
\end{figure*}

\begin{theorem}[Conversion Error Bound]\label{thm:conversionErrorBound}
Given $v$ with $\alpha = DP(v) \leq 22$ and $\beta = DS(v)\leq 15$, let $\varepsilon_i = |v \otimes 10^i - round(v \otimes 10^i)|$ and $\mu_i = |v \otimes 10^i| \times 2^{-52}$, for $i \in \{0, 1, ..., \alpha - 1\}$, we have $\varepsilon_i > \mu_i$ and $\varepsilon_\alpha \leq \mu_\alpha$.
\end{theorem}
\begin{proof}
If $v = 0$, $\alpha = 0$, $\varepsilon_\alpha = 0$ and $\mu_\alpha = 0$, i.e., Theorem~\ref{thm:conversionErrorBound} holds. Next, we prove that Theorem~\ref{thm:conversionErrorBound} holds if $v \neq 0$.

Since $\beta \leq 15$ and $\alpha \leq 22$, for $i \in \{0, 1, ..., \alpha\}$, $v \otimes 10^i$ and $v^{(i)}$ are normal numbers. $|v \otimes 10^i - v^{(i)}|$ means the error introduced by the floating-point arithmetic $v \otimes 10^i$ itself. Let $\mu'_i = ULP(|v \otimes 10^i|)$, as discussed in Section~\ref{subsec:fplayout}, we have \underline{$|v \otimes 10^i - v^{(i)}| \leq \mu'_i$}. For the normal number $v \otimes 10^i$, we have \underline{$\mu'_i \leq |v \otimes 10^i| \times 2^{-52} = \mu_i < 2 \times \mu'_i$}.

We first prove $\varepsilon_\alpha \leq \mu_\alpha$. Since $\beta \leq 15$, according to Theorem~\ref{the:round} (Conversion Correctness), we have $round(v \otimes 10^\alpha) = v^{(\alpha)}$. Therefore, $\varepsilon_{\alpha} = |v \otimes 10^\alpha - v^{(\alpha)}| \leq \mu_\alpha$.

We now prove that for $i \in \{0, 1, ..., \alpha - 1\}$, $\varepsilon_i > \mu_i$.
\textbf{Firstly}, because $i < \alpha$, $v^{(i)}$ is definitely not an integer. According to Theorem~\ref{the:keepbeta} (Decimal Place Invariance), $DS(v^{(i)}) = DS(v) = \beta$ , so we have $|v^{(i)}| = a \times 10^b$, where $1 \leq a < 10$ and $DS(a) = \beta$. We let $DF(a) = (d_0.d_{-1}...d_{-\beta + 1})_{10}$ where $d_{-\beta + 1} \neq $ `0', so we have $\underline{|v^{(i)} - round(v^{(i)})|} \geq min(d_{-\beta + 1}, 10 - d_{-\beta + 1}) \times 10^{b - \beta + 1} \geq 10^{b - \beta + 1} = 10 \times 10^{b} \times 10^{-\beta} > |v^{(i)}| \times 10^{-\beta} \geq \underline{|v^{(i)}| \times 10^{-15}}$.
\textbf{Secondly}, as discussed in Section~\ref{subsec:fplayout}, for the normal number $v^{(i)}$, we have $ULP(|v^{(i)}|) \leq |v^{(i)}| \times 2^{-52} \Rightarrow 1 / ULP(|v^{(i)}|) \geq 1/(|v^{(i)}| \times 2^{-52}) \Rightarrow |v^{(i)} - round(v^{(i)})| / ULP(|v^{(i)}|)$ $> 10^{-15} / 2^{-52} > 4.5 \Rightarrow \\ \underline{|v^{(i)} - round(v^{(i)})| > 4.5 \times ULP(|v^{(i)}|)}$. 
{\bf Thirdly}, suppose $e_{v^{(i)}} = \lfloor log_2|v^{(i)}|\rfloor + 1023$ and $e_{v \otimes 10^i} = \lfloor log_2|v \otimes 10^i|\rfloor + 1023$. For the normal number $v^{(i)}$, we have $2^{e_{v^{(i)}}} \leq |v^{(i)}| < 2^{e_{v^{(i)}} + 1}$ (see Section~\ref{subsec:fplayout}). Since the binary round operation always maps $v^{(i)}$ to the nearest representable floating-point value, if $|v^{(i)}| \in (2^{e_{v^{(i)}} + 1}-2^{e_{v^{(i)}} + 1 - 53}, 2^{e_{v^{(i)}} + 1}) \Leftrightarrow |v^{(i)} - 2^{e_{v^{(i)}} + 1}| < 2^{e_{v^{(i)}} + 1 - 53} \approx 2^{e_{v^{(i)}}} \times 2.22 \times 10^{-16}$, $|v^{(i)}|$ is rounded to the integer $2^{e_{v^{(i)}} + 1}$ (i.e., $v \otimes 10^i = 2^{e_{v^{(i)}} + 1}$, and we also have $round(|v^{(i)}|) = 2^{e_{v^{(i)}} + 1}$). In this case, $e_{v \otimes 10^i} = e_{v^{(i)}} + 1 \Rightarrow ULP(v \otimes 10^i) = 2 \times ULP(v^{(i)})$. But this case never exists for $\beta \leq 15$, because $|v^{(i)} - round(v^{(i)})| = |v^{(i)} - 2^{e_{v^{(i)}} + 1}| > |v^{(i)}| \times 10^{-15} \geq 2^{e_{v^{(i)}}} \times 10^{-15}$, which contradicts $|v^{(i)} - 2^{e_{v^{(i)}} + 1}| < 2^{e_{v^{(i)}} + 1 - 53}$. Therefore, $|v^{(i)}|$ must be mapped to the value with the same exponent with $|v^{(i)}|$, i.e., $e_{v \otimes 10^i} = e_{v^{(i)}} \Rightarrow ULP(|v \otimes 10^i|) = ULP(|v^{(i)}|) = \mu'_i \Rightarrow \underline{|v^{(i)} - round(v^{(i)})| > 4.5 \times \mu'_i}$. We also have $\underline{|v^{(i)} - round(v^{(i)})| \leq 0.5}$ according to the definition of decimal round.
{\bf Finally}, we let $p = v \otimes 10^i - v^{(i)}$, $q = v^{(i)} - round(v^{(i)})$, $r = round(v^{(i)}) - round(v \otimes 10^i)$. Therefore, we have \underline{$|p| \leq \mu'_i$}, \underline{$4.5\times \mu'_i < |q| \leq 0.5$}, \underline{$|r| = 0$ or $1$}, and \underline{$\varepsilon_i = |p + q + r|$}. 
Because $|p| \leq \mu'_i$ and $4.5\times \mu'_i < |q| \leq 0.5$, we have \underline{$0 \leq \mu'_i < 1 / 9$}. We consider three cases. 
{\bf Case~1}: $r = 0$. Considering $|p| \leq \mu'_i \Leftrightarrow -\mu'_i \leq p \leq \mu'_i$ and $|q| > 4.5 \times \mu'_i$, {\bf Case~1.1}: if $q \geq 0$, we have $q > 4.5 \times \mu'_i \Rightarrow p+q > 3.5 \times \mu'_i \geq 0 \Rightarrow \varepsilon_i = p + q > 2 \times \mu'_i > \mu_i$; {\bf Case~1.2}: if $q < 0$, we have $q < -4.5 \times \mu'_i \Rightarrow p + q < -3.5 \times \mu'_i \leq 0 \Rightarrow \varepsilon_i = - (p + q) > 2 \times \mu'_i > \mu_i$.
{\bf Case~2}: $r = 1$. {\bf Case~2.1}: if $q \geq 0$, as discussed in Case~1.1, we have $\varepsilon_i = p + q + 1 > 3.5 \times \mu'_i + 1 > 2 \times \mu'_i > \mu_i$; {\bf Case~2.2}: if $q < 0$, since $|q| \leq 0.5$, we have $-0.5 \leq q < 0$. Considering $|p| \leq \mu'_i$ and $0 \leq \mu'_i < 1/9$, we have $p + q + 1 \geq -\mu'_i + (-0.5) + 1 = 0.5 - \mu'_i > 0.5 - 1/9 = 7 / 18 \Rightarrow \varepsilon_i = |p + q + 1| > 2/9 > 2 \times \mu'_i > \mu_i$.
{\bf Case~3}: $r = -1$. Since $|p| \leq \mu'_i < 1/9$ and $|q| \leq 0.5$, we have $p + q < 1/9 + 0.5 = 11/18$ and $1 - 2 \times \mu'_i > 1 - 2 \times 1/9 = 7 / 9 \Rightarrow p + q < 1 - 2 \times \mu'_i \Rightarrow p + q - 1 < -2 \times \mu'_i \leq 0 \Rightarrow \varepsilon_i = |p + q - 1| > 2 \times \mu'_i > \mu_i$.
\end{proof}

\vspace{-10pt}
\begin{algorithm}[h]
\small
	\caption{$DPAndDSCalculation(v)$}\label{alg:DPAndDSCalculation}
	\lIf(\tcp*[h]{Special case}\label{alg:DPAndDSCalculation:specialCase}){$v = 0$}{\Return{$(0, 0)$}}
	$\alpha \leftarrow 0; \beta \leftarrow \alpha + \lfloor log_{10}|v|\rfloor + 1$; \tcp{Equ.~(\ref{equ:relation})}
	\While{$\beta \leq 15$ {\rm and} $\alpha \leq 22$}{
		\If{$|v \otimes 10^\alpha - round(v \otimes 10^\alpha)| \leq |v \otimes 10^\alpha| \times 2^{-52}$}{
            \If{$round(v \otimes 10^{\alpha}) \div 10^\alpha \neq v$}{
                $\alpha \leftarrow 23$; $\beta \leftarrow 16$\;
            }
            \Break\;
        }
		$\alpha \leftarrow \alpha + 1$; $\beta \leftarrow \beta + 1$\;
	}
	\Return{$(\alpha, \beta)$;}\tcp{We regard $\alpha > 22$ or $\beta > 15$ as Exception}
\end{algorithm}
\vspace{-10pt}


Therefore, we can enumerate $i$ from 0 until $\varepsilon_i \leq \mu_i$, at which point $ i = \alpha$. Algorithm~\ref{alg:DPAndDSCalculation} shows the pseudocode of the calculation of $\alpha$ and $\beta$ of $v$. In Line~\ref{alg:DPAndDSCalculation:specialCase}, if $v = 0$, we set $\alpha = \beta = 0$, and return them directly. Otherwise, we enumerate $\alpha$ from 0, until we encounter the exception (i.e., $\beta > 15$ or $\alpha > 22$) or find the results. We design another method for the exception, so the exception will not affect the compression. The loop is executed at most $15$ times, leading to a time complexity of $\mathcal{O}(15)$. Its space complexity is $\mathcal{O}(1)$ obviously.

\subsubsection{Digit Transformation}\label{subsubsec:digittransform}

Intuitively, given a chunk of values $V = \{v_1, v_2, ..., v_n\}$, we can calculate their decimal places $\mathcal{A} = \{\alpha_1, \alpha_2, ..., \alpha_n\}$ and decimal significands $\mathcal{B} = \{\beta_1, \beta_2, ..., \beta_n\}$ based on Algorithm~\ref{alg:DPAndDSCalculation}. For $1 \leq i \leq n$, if $\alpha_i \leq 22$ and $\beta_i \leq 15$, we convert $v_i$ into $round(v_i \otimes 10^{\alpha_i})$, and record $\alpha_i$ for the recovery of $v_i$. However, if we record every $\alpha_i$, it would occupy too much space, which degrades the compression ratio. Since in a floating-point dataset, the decimal places of values tend to be similar~\cite{li2023elf, Elf+, li2025adaptive}, {\em Falcon} regards the decimal place of all $v_i$ as $\alpha_{max} = max\{\alpha_1, \alpha_2, ..., \alpha_n\}$, and record $\alpha_{max}$ only. Based on $\alpha_{max}$, we calculate a new decimal significand set $\hat{\mathcal{B}} = \{\hat{\beta_1}, \hat{\beta_2}, ..., \hat{\beta_n}\}$, where $\hat{\beta_i} = \alpha_{max} + \lfloor log_{10}{|v_i|} \rfloor + 1$ if $v_i \neq 0$, or $\hat{\beta_i} = 0$ if $v_i = 0$. Let $\hat{\beta_{max}} = max\{\hat{\beta_1}, \hat{\beta_2}, ..., \hat{\beta_n}\}$ (we also have $\hat{\beta_{max}} = \alpha_{max} + \lfloor log_{10}{v_{max}} \rfloor + 1$ if $v_{max} \neq 0$, or $\hat{\beta_{max}} = 0$ if $v_{max} = 0$, where $v_{max} = max\{|v_1|, |v_2|, ..., |v_n|\}$). Next, we elaborate on the digit transformation from two cases.



{\bf Case~1}: If $\alpha_{max} \leq 22$ and $\hat{\beta_{max}} \leq 15$, as shown in Fig.~\ref{fig:IntegerConversion}(a), we first convert $V = \{v_1, v_2, ..., v_n\}$ into $G = \{g_1, g_2, ..., g_n\}$, where $g_i = round(v_i \otimes 10^{\alpha_{max}})$ for $1 \leq i \leq n$. The recoverability is guaranteed by the following theorem.

\begin{theorem}[Conversion Recoverability]
If $\alpha_{max} \leq 22$ and $\hat{\beta_{max}} \leq 15$, we have $v_i = g_i \div 10^{\alpha_{max}}$ for $1 \leq i \leq n$.
\end{theorem}
\begin{proof}
%
If $v_i = 0$, we have $g_i = 0$, i.e., $v_i = g_i \div 10^{\alpha_{max}}$.

If $v_i \neq 0$, as $\alpha_i \leq \alpha_{max}$, $v_i^{(\alpha_{max})}$ must be an integer. Since $\hat{\beta_i} = \alpha_{max} + \lfloor log_{10}{|v_i|} \rfloor + 1$ and $\beta_i = \alpha_i + \lfloor log_{10}{|v_i|} \rfloor + 1$, we have $\hat{\beta_i} - \beta_i = \alpha_{max} - \alpha_i$. Because $DS(v_i^{(\alpha_i)}) = DS(v_i) = \beta_i$, and $v_i^{(\alpha_{max})} = v_i^{(\alpha_i)} \times 10^{\alpha_{max} - \alpha_i}$ (i.e., $v_i^{(\alpha_{max})}$ shifts the decimal point of $v^{(\alpha_i)}$ to the right by $\alpha_{max} - \alpha_i$ positions), we have $DS(v_i^{(\alpha_{max})}) = \beta_i + (\alpha_{max} - \alpha_i) = \hat{\beta_i} \leq 15$. By regarding $v_i^{(\alpha_{max})}$ as $v^{(\alpha)}$ in the proof of Theorem~\ref{the:round}, we can prove that $g_i = round(v_i \otimes 10^{\alpha_{max}}) = v_i^{(\alpha_{max})}$. Since $g_i$ and $10^{\alpha_{max}}$ are both representable values without any binary round error, as discussed in Theorem~\ref{the:recoverable}, we have $v_i = g_i \div 10^{\alpha_{max}}$.
\end{proof}

For example, as shown in Fig.~\ref{fig:IntegerConversion}(a), suppose $\alpha_{max} = 2$, although $DP(-1.2) = 1$, we still convert it into -120 by $round(-1.2 \otimes 10^2)$.

We further transform $G = \{g_1, g_2, ..., g_n\}$ into $Z = \{z_1, z_2, ..., z_n\}$, where $z_1 = g_1$ and $z_i = Zigzag(g_i - g_{i-1})$ for $i \in (1, n]$. Here, if $1 < i \leq n$, $Zigzag(x_i) = (x_i << 1) \oplus (x_i >> 63)$ is the Zigzag encoding~\cite{Zigzag2001} function to ensure $z_i$ to be positive, where $\oplus$ is the bitwise XORing operation, and $x_i = g_i - g_{i - 1}$ (a.k.a. $\Delta$ operation) is based on the assumption that two consecutive values in a dataset usually do not vary significantly, thereby making $z_i$ tend to be small. By doing this, we make each $z_i$ contain as many leading zeros as possible to facilitate our subsequent encoding.

{\bf Case~2}:~If $\alpha_{max} > 22$ or $\hat{\beta_{max}} > 15$, since $v_i$ may not be able to recover from $round(v_i \otimes 10^{\alpha_{max}})$, we propose to adopt another conversion method. Specifically, as shown in Fig.~\ref{fig:IntegerConversion}(b), we first convert $V = \{v_1, v_2, ..., v_n\}$ into $G = \{g_1, g_2, ..., g_n\}$, where $g_i = Zigzag(BinLong(v_i))$, and then transform $G = \{g_1, g_2, ..., g_n\}$ into $Z = \{z_1, z_2, ..., z_n\}$ as the Case 1 does, i.e., $z_1 = g_1$ and $z_i = Zigzag(g_i - g_{i-1})$ for $i \in (1, n]$. Here, $BinLong(v_i)$ means that we interpret the underlying floating-point bits of $v_i$ as a long integer directly (e.g., {\tt Double.doubleToLongBits} in Java). 

\sstitle{Discussion} In Case 2, we perform an extra Zigzag operation when getting $g_i$. Let $g'_i = BinLong(v_i)$ and $Zigzag(g'_i) \approx 2 \times |g'_i|$. In this case, if there are two consecutive values $v_t$ and $v_{t - 1}$ with different signs, since $|g'_t|$ and $|g'_{t-1}|$ are usually very large and $|g'_t| \approx |g'_{t-1}|$, $z'_t = Zigzag(g'_t - g'_{t-1}) \approx 4 \times |g'_t|$ will result in an extremely large value that dominates $Z$, i.e., $z'_t$ is the largest value in $Z$. Reducing $z_t$ plays a vital role in high-performance encoding (detailed in Section~\ref{subsec:encoding}). We can roughly prove that if $|g'_t| < |g'_{t-1}| < 3 \times |g'_t|$ or $g'_{t-1} < |g'_{t}| < 3 \times |g'_{t-1}|$, $z_t = Zigzag(Zigzag(g'_t) - Zigzag( g'_{t-1}))$ is usually smaller than $z'_t$ (i.e., $z'_t > z_t \xLeftrightarrow{\approx} 2 \times |g'_t - g'_{t-1}| > 4 \times |g'_t + g'_{t - 1}| \Leftrightarrow |g'_t| + |g'_{t-1}| > 2 \times |g'_t + g'_{t - 1}| \Leftrightarrow |g'_t| < |g'_{t-1}| < 3 \times |g'_t|$ or $g'_{t-1} < |g'_{t}| < 3 \times |g'_{t-1}|$). 
For example, considering the third value in Fig.~\ref{fig:IntegerConversion}(b), if we do not perform the first Zigzag operation, we will get $z'_3 =$ 18,433,351,846,175,465,127, which is much larger than 77,963,581,577,303,054. 

In Case 1, we do not perform the extra Zigzag operation when getting $g_i$, as $g_i = round(v_i \otimes 10^{\alpha_{max}})$ is generally much smaller compared to that in Case~2. Even if $v_t$ and $v_{t-1}$ have different signs, $z_t = Zigzag(g_t - g_{t - 1})$ usually does not dominate $Z$. Inversely, the extra Zigzag operation usually increases the maximum value in $Z$.

Because the values in a dataset tend to have similar decimal places and decimal significands, when compressing a single dataset, all threads in GPUs tend to handle the same case. That is, the two distinct processing branches result in negligible warp divergence.

\sstitle{Summary} Combining Case 1 and Case 2 together, we have:
\begin{equation}\label{equ:z1}
g_i = \left\{  
	\begin{array}{ll}
		round(v_i \otimes 10^{\alpha_{max}}) \quad {\rm if}\ \alpha_{max} \leq 22 \ {\rm and}\ \hat{\beta_{max}} \leq 15\\
		Zigzag(BinLong(v_i)) \quad {\rm if}\ \alpha_{max} > 22 \ {\rm or}\ \hat{\beta_{max}} > 15
	\end{array}
\right. 
\end{equation}%
\begin{equation}\label{equ:z2}
z_i = \left\{  
	\begin{array}{ll}
		g_1 \quad {\rm if}\ i = 1\\
		Zigzag(g_i - g_{i-1}) \quad {\rm if}\ i > 1
	\end{array}
\right. 
\end{equation}

\subsection{Adaptive Bit Plane Encoding}\label{subsec:encoding}

\begin{figure}[t]
  \centering
  \includegraphics[width=3.4in]{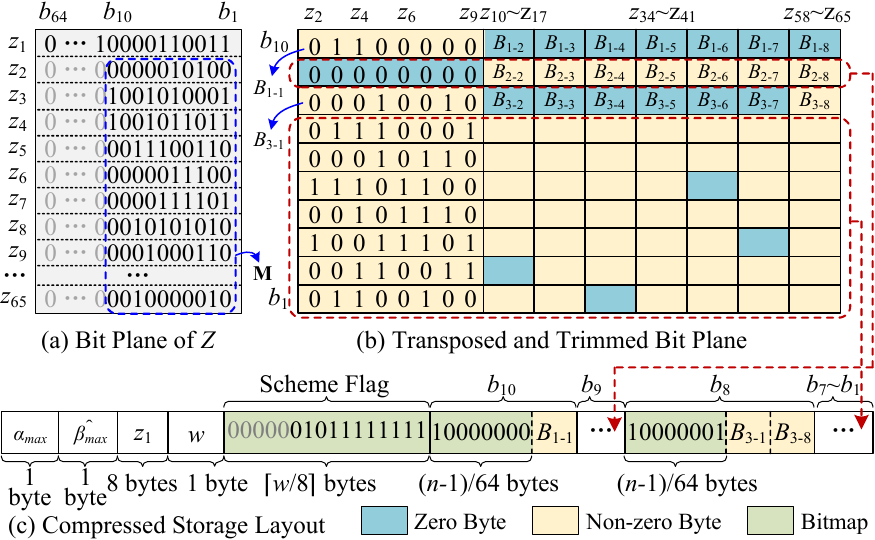}
  \vspace{-15pt}
  \caption{Adaptive Bit Plane Encoding ($n = 65$ in the example).}
  \label{fig:BitPlaneEncoding}
  \vspace{-10pt}
\end{figure}

In this step, we encode $Z=\{z_1, z_2, ..., z_n\}$, where $z_i \in Z$ is a long integer that occupies 64 bits, and $n$ is the number of values in a chunk. Since $z_i \in Z$ is positive and usually small, $z_i$ tends to contain many leading zeros. The number of leading zeros of $z_i$ is $lead_i = 64 - (\lfloor log_2{z_i} \rfloor + 1)$. We let $lead_{min} = min\{lead_2, lead_3, ..., lead_{65}\}$. Thus, we trim the first $lead_{min}$ leading zeros of $z_j \in Z\setminus\{z_1\}$, forming a new bit plane matrix $\mathbf{M} \in \{0, 1\}^{(n-1)\times w}$, where $w = 64 - lead_{min}$ is the {\bf bit width}. Note that we do not consider $z_1$ here because we store $z_1$ as it is. Fig.~\ref{fig:BitPlaneEncoding}(a) gives an example, where $n = 65$ and $w = 10$.

As each row of $\mathbf{M}$ may not be byte-aligned, we get the transposed bit plane matrix $\mathbf{M}^{\rm T} \in \{0, 1\}^{w \times (n-1)}$, since we can ensure each row of $\mathbf{M}^{\rm T}$ is byte-aligned by setting $n = k \times 64 + 1$, where $k$ is a positive integer (in our implementation, $n = 1025$). Suppose $B_{i\text{-}j}$ is the $j$-th byte in the $i$-th row of $\mathbf{M}^{\rm T}$. As shown in Fig.~\ref{fig:BitPlaneEncoding}(b), if all bits in $B_{i\text{-}j}$ are 0s, we call $B_{i\text{-}j}$ {\bf zero byte} (e.g., $B_{2\text{-}1}$ and $B_{3\text{-}2}$ marked in blue); otherwise, we call it {\bf non-zero byte} (e.g., $B_{1\text{-}1}$ and $B_{2\text{-}2}$ marked in yellow). We devise two storage schemes for each row of $\mathbf{M}^{\rm T}$. 1)~{\bf Sparse Storage}. We utilize a bitmap with a size of $(n - 1) / 64$ bytes to indicate the byte types in the row (`1' represents non-zero byte while `0' stands for zero byte), sequentially followed by the non-zero bytes themselves in the row. As shown in Fig.~\ref{fig:BitPlaneEncoding}(c), the bitmap of the first row (i.e., $b_{10}$) of $\mathbf{M}^{\rm T}$ is ``$10000000$'', representing there is only one non-zero byte (i.e., the first byte). 2)~{\bf Dense Storage}. We store the bytes in the row as they are. Suppose $\lambda_i$ be the number of zero bytes in the $i$-th row, then the cost of Sparse Storage is $(n-1) / 64 + ((n - 1) / 8 - \lambda_i)$ bytes, while the cost of Dense Storage is $(n - 1) / 8$ bytes. We propose an adaptive storage scheme selection with the minimum storage cost. Specifically, if $(n-1) / 64 + ((n - 1) / 8 - \lambda_i) < (n - 1) / 8 \Leftrightarrow \lambda_i > (n - 1) / 64$, we adopt the Sparse Storage scheme; otherwise, we adopt the Dense Storage scheme. 
The advantages of our design are twofold. First, due to the existence of outliers, there are commonly many zero bytes in the first few rows. Our Sparse Storage scheme can well handle these outliers and achieve good compression ratio. Second, the zero bytes mainly appear in the first few rows, i.e., the threads in GPUs do not execute branches randomly, which alleviates the warp divergence issue.

Fig.~\ref{fig:BitPlaneEncoding}(c) presents the underlying compressed storage layout. We first store $\alpha_{max}$ and $\hat{\beta_{max}}$ with 1 byte, respectively, and then store $z_1$ with 8 bytes, followed by 1 byte of $w$. Because $w \leq 64$, 1 byte (i.e., 8 bits) is enough for $w$. The following $\lceil w / 8 \rceil$ bytes are the storage scheme flags, where the last $w$ bits respectively indicate the adopted storage scheme of each row of $\mathbf{M}^{\rm T}$ (`0' represents the sparse storage scheme while `1' is the dense storage scheme). The actual compressed data of each row are stored in sequence. 

\subsection{Compressed Size Sync and Chunk Writing}\label{subsec:sync}

The compressed data from each thread are stored in the local memory of the GPU, and need to be merged into a single stream in the global memory. One basic solution is to reserve a portion of the global memory for each thread, allowing each thread to copy its compressed data into its corresponding space. However, this solution presents two issues: 1)~Determining the size of the reserved space is challenging. If the reserved space is too small, it cannot accommodate the compressed data; if it is too large, it leads to wasteful utilization of GPU global memory. 2)~An additional merging operation is required, thereby affecting compression throughput. 


Therefore, this paper proposes to first synchronize the compressed chunk sizes among threads. Based on the compressed chunk sizes, we calculate the offset of each compressed chunk data, and thus each thread is able to copy its compressed chunk data directly into the right address in the global memory. We employ the decoupled look-back strategy~\cite{huang2024cuszp2, merrill2016decoupled} to accelerate this process.

When all threads have finished writing their chunks, the compression of one data batch is complete  (i.e., the CmpKernel in Fig.~\ref{fig:DifferentSchedulers}). At this point, the compressed data residing in the global memory consists of two components: a compressed chunk size array and the actual compressed chunk data, which form the complete compressed data of the batch.

\subsection{Summary of {\em Falcon} Compression}
Algorithm~\ref{alg:CmpKernelAsync} depicts the pseudocode of {\em Falcon}'s CmpKernel (which is called by Algorithm~\ref{alg:EventDrivenScheduler}). Most lines of Algorithm~\ref{alg:CmpKernelAsync} are self-explanatory. Note that the code segments in Lines~2-15 are executed in parallel across different GPU threads. 

\begin{algorithm}[htb]
\small
	\caption{$CmpKernelAsync(gInArr, gOutArr, streamId)$}\label{alg:CmpKernelAsync}
	Divide $gInArr$ into multiple chunks, each containing $n$ values\;
	\Parallel{{\rm each GPU thread processes a chunk} $V = \{v_1, v_2, ..., v_n\}$}{
		\tcc{Integer Conversion}
		$\alpha_{max} \leftarrow max\{DP(v_1), DP(v_2), ..., DP(v_n)\}$\;
		$v_{max} \leftarrow max\{|v_1|, |v_2|, ..., |v_n|\}$\;
		\lIf{$v_{max} = 0$}{$\hat{\beta_{max}} \leftarrow 0$}
		\lElse{$\hat{\beta_{max}} \leftarrow \alpha_{max} + \lfloor log_{10}{v_{max}} \rfloor + 1$}
		Calculate $Z = \{z_1, z_2, ..., z_n\}$ based on Equ.~(\ref{equ:z1}) and Equ.~(\ref{equ:z2})\;
		\tcc{Adaptive Bit Plane Encoding}
		Construct $\mathbf{M}$ based on $Z \setminus \{z_1\}$, and then get $\mathbf{M}^{\rm T}$\;
		\For{{\rm each row} $row_i$ {\rm in} $\mathbf{M}^{\rm T}$}{
			$\lambda_i \leftarrow$ the number of zero bytes in $row_i$\;
			\lIf{$\lambda_i > (n-1)/64$}{Encode $row_i$ with Sparse Storage}
			\lElse{Encode $row_i$ with Sparse Storage}
		}
		Construct the compressed chunk data shown in Fig.~\ref{fig:BitPlaneEncoding}\;
		\tcc{Compressed Size Sync and Chunk Writing}
		Synchronize the compressed chunk size and get result offset\;
		Write the compressed chunk data into $gOutArr$ with the offset\;
	}
	Copy the compressed chunk size array into $gOutArr$\;
\end{algorithm}

\section{{\em Falcon} Decompresion}\label{sec:decompression}

As shown in Fig.~\ref{fig:framework}(b), {\em Falcon} decompression is essentially the inverse of its compression. Specifically, the compressed chunk size array is first read directly from the compressed data, based on which the compressed chunk offset array is then computed and broadcast in the global memory, facilitating parallel access by GPU threads. Subsequently, GPU threads concurrently read the compressed chunk data and perform adaptive bit-plane decoding to reconstruct a set of integer values. The original floating-point values are then recovered by applying the inverse operations of Equ.~(\ref{equ:z1}) and Equ.~(\ref{equ:z2}). Finally, the decompressed floating-point values are written in parallel to their corresponding locations in the global memory. Note that since the decompressed size of all chunks is identical and known a priori, we adopt Pre-Allocation Scheduler shown in Fig.~\ref{fig:DifferentSchedulers}(a) during decompression.

\section{Experimental Evaluation}\label{sec:experiment}
\subsection{Datasets and Experiment Settings}\label{subsec:setup}


\subsubsection{Datasets.} Our evaluation is performed on 12 diverse floating-point datasets, the details of which are summarized in Table~\ref{tab:datasets}. Specifically, we select 9 large-volume real-world datasets covering financial, geospatial and scientific domains. Furthermore, we include 3 synthetic datasets generated by the Time Series Benchmark Suite (TSBS)~\cite{serf_tsbs_2023} and NYX~\cite{NYX_Almgren_2013} to assess the performance on controlled and simulated data patterns from IoT, system monitoring, and cosmological simulations, respectively.

\begin{table}[t] 
\centering 
\caption{Details of Datasets.}
\label{tab:datasets}
\setlength{\tabcolsep}{1pt} 
\resizebox{3.45in}{!}{ 
\begin{tabular}{@{}lcccl@{}}
\toprule
\textbf{Dataset}       & \textbf{\#Records} &\textbf{${\beta}_{avg}$} &\textbf{${\beta}_{max}$} & \textbf{Description} \\ 
\midrule
\multicolumn{5}{l}{\textit{Real-World Datasets}} \\ 
\midrule
Air-pressure (AP)~\cite{neon_pressure_2022}    & 13356090      & 8   &8    & Barometric pressure        \\
City-temp (CT)~\cite{Kaggle_city_temp2019}        & 57727468      & 3   &5    & Main cities' temperature        \\
Gas-sensor (GS)~\cite{gas_sensor_2015}                  & 67332177    & 6     &7    & Ethylene CO sensor data        \\
JaneStreet-market (JM)~\cite{JaneStreetKaggle}       & 71935770     &8    &9    & Jane street market prices        \\
Stocks-price (SP)~\cite{elf_stocks_2023}      & 11321684     & 4    &6    & Stock price in 3 countries         \\
Solar-wind (SW)~\cite{solar_wind_Kaggle}       & 104139947     & 3   &6    & NASA solar wind data        \\ 
Taxi-amount (TA)~\cite{NYCTaxi2021}      & 47248845     & 3   & 8    & Taxi order fare amount        \\
Taxi-position (TP)~\cite{NYCTaxi2021}      & 47248844     & 16  & 17    & Taxi order pickup position         \\

Wind-speed (WS)~\cite{neon_windspeed_2022}       & 20860255     & 3   & 7     & Wind speed data        \\ \midrule
\multicolumn{5}{l}{\textit{Synthetic Datasets}} \\ 
\midrule
NYX~\cite{NYX_Almgren_2013}      & 9437184       &8   &9   & Cosmological hydrodynamics        \\  
Sim-Memory (SM)~\cite{serf_tsbs_2023}       & 51840000      & 10  &11   & Simulated memory metrics    \\     
Sim-Truck (ST)~\cite{serf_tsbs_2023}       & 34422759      & 8  &9   & Simulated  truck position       \\      
\bottomrule
\end{tabular}
} 
\end{table}

\subsubsection{Baselines.} 
We compare \textit{Falcon} against 10 representative competitors, including 3 state-of-the-art CPU-based floating-point lossless compressors (i.e.,  Elf~\cite{li2023elf}, ALP~\cite{afroozeh2023alp}, and Elf*~\cite{li2025adaptive}), 1 GPU-based lossless compressor designed specifically for scientific data (i.e., ndzip~\cite{9910073, knorr2021ndzip}), and 4 GPU-based general-purpose compression algorithms (i.e., Bitcomp, LZ4, Snappy, GDeflate) that are extracted from NVIDIA's nvCOMP~\cite{nvcomp2023} library. To enable a fair assessment on the same hardware, we also transplant the most 2 state-of-the-art CPU-based floating-point compressors to the GPU (i.e., GPU:Elf*, GPU:ALP), allowing us to evaluate the core algorithmic innovations of {\em Falcon} under equitable conditions. The pioneering GPU-based floating-point compression methods~\cite{o2011floating, yang2015mpc} are not compared, because they have been proven to perform worse than the more advanced ndzip~\cite{9910073, knorr2021ndzip}. The method~\cite{jamalidinan2025floating} is also excluded from our comparison because it contains a clustering-based preprocessing step, followed by GPU-enabled generic compression. The clustering-based preprocessing step is not amenable to GPU acceleration.

\subsubsection{Metrics.} We evaluate the performance by 3 metrics: 1)~\textbf{compression ratio}, defined as the ratio of the compressed data size to the original one, 2)~\textbf{compression throughput (GB/s)}, which means the original data size divided by the total processing time, and 3)~\textbf{decompression throughput (GB/s)}, which indicates the decompressed data size divided by the total processing time.

\begin{table}[t!]
  \centering
  \captionsetup{skip=5pt} 
  \caption{Compression Ratio (The best results are in \textbf{bold}, and the suboptimal results are in underlined. The same below).}
  \label{tab:compression_ratio}
  \setlength{\tabcolsep}{1.6pt} 
  \resizebox{3.45in}{!}{  
  \begin{tabular}{@{}c|cccccccc|ccc@{}}
\bottomrule
\multirow{2}{*}{\textbf{\makecell[c]{Data\\Set}}} & \multicolumn{8}{c|}{\textbf{GPU}} & \multicolumn{3}{c}{\textbf{CPU}}                                         \\ \cline{2-12} 
       & {\em Falcon}  & ndzip & Lz4  & Bitc. & GDef. & Snap. & Elf* & ALP  & Elf  & Elf*      & ALP\\ \hline
\textbf{AP}                       & \textbf{0.139} & 0.990 & 0.376 & 0.687 & 0.862 & 0.386 & 0.151 & 0.258 & 0.163 & \underline{0.150} & 0.257 \\
\textbf{CT}                       & \textbf{0.096} & 0.989 & 0.388 & 0.882 & 0.472 & 0.345 & 0.158 & \underline{0.133} & 0.189 & 0.160 & \underline{0.133} \\
\textbf{GS}                       & \textbf{0.204} & 1.003 & 0.558 & 1.001 & 0.743 & 0.524 & 0.267 & 0.220 & 0.312 & 0.280 & \underline{0.219} \\
\textbf{JM}                       & 0.550 & 0.990 & 0.730 & 0.974 & 0.907 & 0.745 & 0.584 & \textbf{0.529} & 0.637 & 0.585 & \underline{0.536} \\
\textbf{SP}                       & \textbf{0.096} & 0.988 & 0.424 & 0.899 & 0.635 & 0.403 & \underline{0.144} & 0.161 & 0.174 & 0.157 & 0.161 \\
\textbf{SW}                       & \textbf{0.150} & 1.000 & 0.468 & 1.000 & 0.728 & 0.487 & 0.241 & 0.189 & 0.269 & 0.241 & \underline{0.188} \\
\textbf{TA}                       & \textbf{0.207} & 1.000 & 0.414 & 0.977 & 0.590 & 0.401 & 0.286 & 0.223 & 0.318 & 0.282 & \underline{0.222} \\
\textbf{TP}                       & 0.765 & 1.000 & 0.607 & 0.998 & \textbf{0.536} & 0.569 & 0.538 & 0.830 & 0.630 & \underline{0.537} & 0.830 \\
\textbf{WS}                       & \textbf{0.148} & 0.998 & 0.430 & 0.957 & 0.685 & 0.448 & 0.229 & \underline{0.175} & 0.258 & 0.231 & \underline{0.175} \\
\textbf{NYX}                      & 0.484 & 0.996 & 0.936 & 1.001 & 0.946 & 0.987 & 0.491 & \underline{0.480} & 0.539 & 0.493 & \textbf{0.479} \\
\textbf{SM}                       & \underline{0.436} & 1.000 & 0.943 & 1.001 & 0.939 & 0.979 & 0.529 & \textbf{0.433} & 0.589 & 0.527 & \textbf{0.433} \\ 
\textbf{ST}                       & \underline{0.312} & 1.000 & 0.668 & 0.987 & 0.791 & 0.760 & 0.424 & \textbf{0.310} & 0.478 & 0.422 & \textbf{0.310} \\ \hline
\textbf{Avg.}                     & \textbf{0.299} & 0.996 & 0.579 & 0.947 & 0.736 & 0.586 & 0.337 & \underline{0.329} & 0.380 & 0.339 & \underline{0.329} \\ 
\toprule
\end{tabular}
}
\end{table}

\subsubsection{Settings.}
The number of streams in the pipeline is set to 16 by default. We fix the number of floating-point values in a data chunk to be 1025, which is similar to that of Elf~\cite{li2023elf} and ALP~\cite{afroozeh2023alp}. The default batch size (i.e., the number of floating-point values in a batch) is set to $1025 \times 1024 \times 4$. The general compression competitors are set to their default compression level. All experiments are conducted on a workstation equipped with an NVIDIA GeForce RTX\textsuperscript{TM} 5080 GPU, an Intel\textsuperscript{\textregistered} Core\textsuperscript{TM} i7-14700KF CPU, and 32~GB of DDR4 host memory. The software environment consists of Ubuntu 24.04.2 LTS, running with CUDA Toolkit~12.9 and NVIDIA Driver~575.51.03.

\subsection{Overall Performance Comparison}\label{subsec:overallcmp}

\subsubsection{Compression Ratio.} With regard to the compression ratio, we have the following observations from Table~\ref{tab:compression_ratio}.

{\em Falcon} exhibits improved compression performance compared to state-of-the-art CPU-based floating-point data compression algorithms. On average, {\em Falcon} achieves a relative compression ratio improvement of $(0.38-0.299)/0.38=21\%$ over Elf, 11\% over Elf* and 9\% over ALP. This advantage is also observed when comparing {\em Falcon} to the GPU-ported versions of Elf* and ALP, demonstrating the effectiveness of {\em Falcon}'s GPU implementation. We note that, on the synthetic datasets (i.e., NYX, SM and ST), ALP achieves a slightly better compression ratio than {\em Falcon}. The probable cause is that, the synthetic datasets commonly have  limited data ranges, which are friendly to the FOR (Frame Of Reference) operation of ALP. We also observe that the TP dataset (characterized by a generally $\hat{\beta}_{max} > 15$) exhibits a large number of trailing zeros, which has a unique advantage for XOR-based algorithms such as Elf*.

Furthermore, {\em Falcon} demonstrates a clear advantage in compression ratio over other GPU-based compression algorithms across most datasets. Specifically, {\em Falcon} achieves average relative compression ratio improvements of 70\%, 48\%, 68\%, 59\%, and 49\% over ndzip, Lz4, Snappy, GDeflate and Bitcomp across all datasets, respectively. These results highlight the effectiveness of our GPU implementation and the superiority of {\em Falcon} in terms of compression ratio. Similarly, byte-oriented algorithms like Lz4, Snappy, and GDeflate are well-suited for the TP dataset due to its abundance of trailing zeros.

\begin{table}[t!]
  \centering
  \captionsetup{skip=5pt} 
  \caption{Compression Throughput (GB/s).}
\label{tab:e2eCompThroughput}
  \setlength{\tabcolsep}{2pt} 
  \resizebox{\columnwidth}{!}{  
  \begin{tabular}{@{}c|cccccccc|ccc@{}}
\bottomrule
\multirow{2}{*}{\textbf{\makecell[c]{Data\\Set}}} & \multicolumn{8}{c|}{\textbf{GPU}} & \multicolumn{3}{c}{\textbf{CPU}}\\ \cline{2-12} 
       & {\em Falcon}  & ndzip & Lz4  & Bitc. & GDef. & Snap. & Elf* & ALP  & Elf  & Elf*      & ALP\\ \hline

\textbf{AP}                       & \textbf{8.32} & 1.87 & 1.58 & 3.34 & \underline{3.79} & 0.10 & 0.86 & 3.06 & 0.15 & 0.09 & 0.71 \\
\textbf{CT}                       & \textbf{15.05} & 1.93 & 1.63 & 3.14 & \underline{5.31} & 0.30 & 0.87 & 5.04 & 0.14 & 0.07 & 0.76 \\
\textbf{GS}                       & \textbf{13.23} & 1.92 & 1.52 & 2.78 & 4.50 & 0.19 & 0.82 & \underline{4.64} & 0.12 & 0.06 & 0.75 \\
\textbf{JM}                       & \textbf{10.27} & 1.96 & 1.43 & 2.90 & \underline{4.21} & 0.86 & 0.74 & 3.78 & 0.11 & 0.05 & 0.68 \\
\textbf{SP}                       & \textbf{8.23} & 1.91 & 1.55 & 2.59 & \underline{4.33} & 0.07 & 0.85 & 2.89 & 0.13 & 0.07 & 0.70 \\
\textbf{SW}                       & \textbf{16.19} & 1.93 & 1.59 & 2.88 & 4.62 & 0.44 & 0.84 & \underline{4.74} & 0.12 & 0.06 & 0.74 \\
\textbf{TA}                       & \textbf{13.29} & 1.89 & 1.63 & 2.83 & \underline{4.93} & 0.22 & 0.82 & 4.74 & 0.11 & 0.06 & 0.76 \\
\textbf{TP}                       & \textbf{7.61} & 1.94 & 1.51 & 2.74 & \underline{5.06} & 0.16 & 0.76 & 3.55 & 0.13 & 0.06 & 0.68 \\
\textbf{WS}                       & \textbf{10.74} & 1.85 & 1.53 & 2.68 & \underline{4.50} & 0.10 & 0.84 & 3.96 & 0.13 & 0.07 & 0.76 \\
\textbf{NYX}                      & \textbf{4.91} & 1.91 & 1.49 & 2.29 & \underline{3.75} & 0.13 & 0.75 & 2.16 & 0.11 & 0.05 & 0.68 \\
\textbf{SM}                       & \textbf{10.42} & 1.91 & 1.31 & 2.78 & \underline{4.14} & 0.57 & 0.76 & 2.21 & 0.12 & 0.06 & 0.74 \\
\textbf{ST}                       & \textbf{11.56} & 1.91 & 1.42 & 2.77 & \underline{4.39} & 0.16 & 0.78 & 4.35 & 0.12 & 0.06 & 0.75 \\ \hline
\textbf{Avg.}                     & \textbf{10.82} & 1.91 & 1.52 & 2.81 & \underline{4.46} & 0.28 & 0.81 & 3.76 & 0.12 & 0.06 & 0.73 \\ 
\toprule
\end{tabular}
}
\end{table}

\begin{table}[t!]
  \centering
  \captionsetup{skip=5pt} 
  \caption{Decompression Throughput (GB/s).}
\label{tab:e2eDecompThroughput}
  \setlength{\tabcolsep}{2pt} 
  \resizebox{\columnwidth}{!}{  
  \begin{tabular}{@{}c|cccccccc|ccc@{}}
\bottomrule
\multirow{2}{*}{\textbf{\makecell[c]{Data\\Set}}} & \multicolumn{8}{c|}{\textbf{GPU}} & \multicolumn{3}{c}{\textbf{CPU}}                                         \\ \cline{2-12} 
       & {\em Falcon}  & ndzip & Lz4  & Bitc. & GDef. & Snap. & Elf* & ALP  & Elf  & Elf*      & ALP\\ \hline
\textbf{AP}                       & \textbf{9.36}      & 1.63      & 2.47      & 1.61 & 3.88  & 0.24 & \underline{5.63}    & 1.45     & 0.23   & 0.18 & 4.07 \\
\textbf{CT}                       & \textbf{15.95}      & 1.69      & 2.58      & 2.06 & 5.49  & 1.03 & \underline{6.27}    & 2.30     & 0.21   & 0.16 & 4.34 \\
\textbf{GS}                       & \textbf{14.36}      & 1.72      & 2.47      & 2.13 & 4.64  & 0.70 & \underline{5.55}    & 2.13     & 0.18   & 0.15 & 4.10 \\
\textbf{JM}                       & \textbf{12.82}      & 1.73      & 2.40      & 2.10 & \underline{4.42}  & 1.67 & 4.38    & 1.86     & 0.15   & 0.11 & 3.57 \\
\textbf{SP}                       & \textbf{8.81}      & 1.67      & 2.49      & 1.58 & 4.02  & 0.25 & \underline{5.11}    & 1.41     & 0.19   & 0.16 & 4.04 \\
\textbf{SW}                       & \textbf{16.94}      & 1.72      & 2.57      & 2.15 & 4.92  & 1.18 & \underline{5.79}    & 2.25     & 0.17   & 0.13 & 4.12 \\
\textbf{TA}                       & \textbf{14.25}      & 1.70      & 2.61      & 1.96 & 5.08  & 1.05 & \underline{5.20}    & 2.09     & 0.15   & 0.12 & 4.08 \\
\textbf{TP}                       & \textbf{10.86}      & 1.71      & 2.44      & 1.99 & \underline{5.37}  & 0.91 & 4.77    & 1.59     & 0.19   & 0.12 & 3.42 \\
\textbf{WS}                       & \textbf{11.63}      & 1.66      & 2.48      & 1.84 & 4.43  & 0.34 & \underline{5.57}    & 1.59     & 0.18   & 0.14 & 4.22 \\
\textbf{NYX}                       & \textbf{7.03}      & 1.65      & 2.14      & 1.84 & 3.30  & 0.37 & \underline{3.78}    & 0.91     & 0.15   & 0.11 & 3.68 \\
\textbf{SM}                       & \textbf{13.01}      & 1.71      & 2.26      & 2.15 & 4.24  & 1.24 & \underline{4.64}    & 1.39     & 0.17   & 0.13 & 3.76 \\ 
\textbf{ST}                       & \textbf{12.86}      & 1.71      & 2.39      & 1.91 & 4.45  & 0.65 & \underline{4.86}    & 1.84     & 0.17   & 0.14 & 4.04 \\ \hline
\textbf{Avg.}                     & \textbf{12.32}      & 1.69      & 2.44      & 1.94 & 4.52  & 0.80 & \underline{5.13}    & 1.73     & 0.18   & 0.14 & 3.95 \\ 
\toprule
\end{tabular}
}
\end{table}

\subsubsection{Compression and Decompression Throughput} \label{subsec:e2ethrought}
Based on the throughput results presented in Table~\ref{tab:e2eCompThroughput} and Table~\ref{tab:e2eDecompThroughput}, we have the following observations.

(1) {\em Falcon} demonstrates significantly higher throughput in both compression and decompression compared to other mainstream GPU-based compression algorithms. During compression, the throughput of {\em Falcon} is consistently 2.4$\times$ to 7.1$\times$ faster than competitors like ndzip and the nvCOMP suite. During decompression, this advantage is even more pronounced, with {\em Falcon} achieving throughput up to 2.7$\times$  faster than the best-performing nvCOMP algorithm (i.e., GDeflate), and 7.3$\times$ faster than ndzip.

(2) When compared against state-of-the-art CPU-based algorithms and their GPU-ported versions, {\em Falcon} maintains a clear performance advantage. During compression, {\em Falcon}'s throughput far exceeds that of the CPU-based algorithms, with its average throughput being $87\times$, $15\times$, $169\times$ higher than Elf, Elf* and SIMD-accelerated ALP, respectively. Moreover, compared to our GPU-ported versions, {\em Falcon}'s compression throughput is $13\times$ higher than GPU:Elf* and $2.9\times$ higher than GPU:ALP. This performance superiority extends to the decompression process, where {\em Falcon} demonstrates a similar advantage. Notably, its decompression throughput is $2.4\times$ higher than its most competitive counterpart GPU:Elf*.



\subsection{Parameter Study}

\begin{figure}
  \centering
  \includegraphics[width=3.5in]{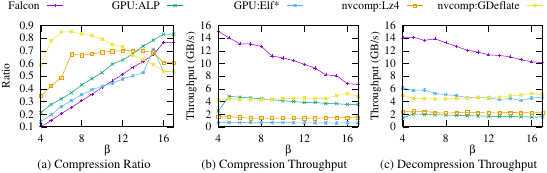}
  \caption{Performance with Different Decimal Significands.}
  \label{fig:6_2Beta}
\end{figure}

\subsubsection{Performance with Different $\beta$}
We observe a strong correlation between the performance of {\em Falcon} and the decimal significand values $\beta$ of the dataset. To investigate this relationship, we conduct a set of experiments on the TP dataset by gradually reducing its decimal significand through string truncation in advance. We then compare {\em Falcon}'s performance against four top-performing algorithms in Section~\ref{subsec:overallcmp}.
The results show that {\em Falcon}'s compression ratio improves as the $\beta$ value increases in Fig.~\ref{fig:6_2Beta}(a). For data with small $\beta$ values, the advantage of our conversion scheme is less pronounced. This is because the binary representation of TP data inherently contains a large number of trailing zeros, which diminishes the relative gains from our method. As illustrated in Fig.~\ref{fig:6_2Beta}(b) and Fig.~\ref{fig:6_2Beta}(c), while {\em Falcon}'s compression and decompression throughputs decrease with larger $\beta$ values, it consistently remains the highest-performing solution among all competitors.

\begin{figure}
  \centering
  \includegraphics[width=\linewidth]{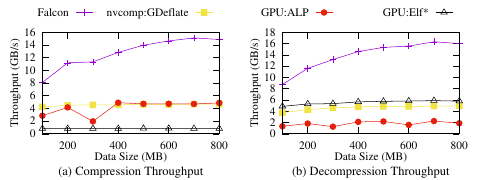}
  \caption{Performance with Different Data Size.}
  \label{fig:6_4_throughput_comparison}
\end{figure}

\subsubsection{Performance with Different Data Sizes}\label{subsec:Scalability_datasize}
We analyze the variation of throughput with increasing data size, as shown in Fig.~\ref{fig:6_4_throughput_comparison}, based on the SW dataset. We do not present the compression ratio, since the compression ratios of all methods keep stable under different total data sizes and the same chunk size. Here we select three competitors with the best throughput for comparison.

The throughput of all competitors saturates quickly. For instance, the throughput of GDeflate stagnates around 4~GB/s, regardless of the data size. This confirms that these methods are architecturally limited and cannot benefit from larger workloads.

In contrast, {\em Falcon} exhibits strong scalability. On the same dataset, as the data size increases from 100 MB to 700 MB, its compression throughput rises from approximately 8~GB/s to over 15~GB/s, representing a growth of more than 1.8$\times$. This indicates that {\em Falcon}'s performance advantage becomes particularly pronounced at larger data scales: its asynchronous pipeline effectively hides latency, making the benefits more significant with an increasing data size.

\begin{table}[t!]
  \centering
  \small
  \caption{Performance with Different Batch Sizes.}
  \label{tab:batchSize}
  \setlength{\aboverulesep}{0pt}
  \setlength{\belowrulesep}{0pt}
  \setlength{\extrarowheight}{0pt}
  \begin{tabular*}{\columnwidth}{@{\extracolsep{\fill}}c|ccccc@{}}
    \toprule
    \textbf{Size($\times$1025$\times$1024)} & 0.5 & 1 & 2 & 4 & 8 \\ 
    \midrule
    \textbf{CT (GB/s)} & 7.26 & 9.37 & 10.52 & \textbf{11.30} & \underline{11.16} \\
    \textbf{DT (GB/s)} & 11.37	& 11.58	& \underline{11.59}	& \textbf{11.61}	& 10.81 \\

    \bottomrule
  \end{tabular*}
\end{table}

\subsubsection{Performance with Different Batch Sizes}
To investigate the impact of batch size on overall performance, we conduct a set of experiments to evaluate its effect on the throughput of {\em Falcon}. The results, presented in Table~\ref{tab:batchSize}, show a clear trend: both the compress throughput (CT) and the decompress throughput (DT) of {\em Falcon} increase as the batch size grows, reaching its peak at a size of 4×1025×1024. This behavior can be explained by a trade-off. A small batch size leads to the under utilization of the GPU's parallel processing capabilities. Conversely, an excessively large batch size can cause significant imbalances in the execution time across different pipeline stages, which in turn degrades throughput. The compression ratio has nothing to do with the batch size given a fixed chunk size, so we do not present the compression ratio here.

\subsection{Ablation Study}

\subsubsection{Analysis of Pipelining Architecture}
Fig.~\ref{fig:6_3_2Ablation} (a) reports the average compression throughput of the three schedulers introduced in Fig.~\ref{fig:DifferentSchedulers} under varying numbers of streams on all datasets.

As the number of streams increases, the compression throughput of both Pre-Allocation Scheduler and Event-Driven Scheduler gradually improves, plateauing after reaching 16 streams. In contrast, Sync-Based Scheduler maintains a constant Compression throughput regardless of stream number. This behavior occurs because a higher stream number enables Pre-Allocation Scheduler and Event-Driven Scheduler to process more streams concurrently. When the number of streams reaches 16, the system's processing capacity becomes saturated. However, Sync-Based Scheduler serializes stream execution due to its synchronous operations, which hinders parallel processing and limits throughput scalability.

Event-Driven Scheduler consistently achieves the highest compression throughput in multi-stream configurations. When there are 16 streams, it achieves over 2$\times$ improvement compared to the single-stream case and outperforms Pre-Allocation Scheduler and Sync-Based Scheduler by 2.6$\times$ and 2$\times$, respectively, demonstrating the effectiveness of our proposed approach. Pre-Allocation Scheduler manifests the lowest compression throughput in all numbers of streams, because it transfers many useless data from GPU to CPU, and performs an extra merging step.

\begin{figure}
  \centering
  \includegraphics[width=3.5in]{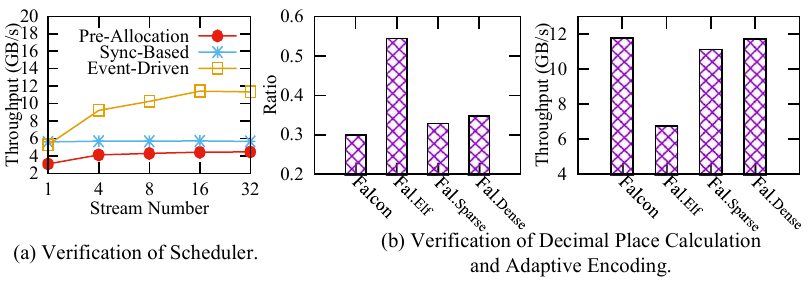}
  \caption{Ablation Study Results.}
  \label{fig:6_3_2Ablation}
\end{figure}

\subsubsection{Analysis of Algorithmic Components}
To further evaluate the contributions of our core algorithmic innovations, we conduct another ablation study by comparing the performance of {\em Falcon} against three variants with key optimizations disabled. Specifically, we compared Falcon with: \textbf{{\em Fal.}\textsubscript{Elf}}, which adopts the imprecise decimal place calculation scheme from Elf; \textbf{{\em Fal.}\textsubscript{Sparse}}, which treats all bit-planes as sparse; and \textbf{{\em Fal.}\textsubscript{Dense}}, which treats all bit-planes as dense. We do not compare the string-analysis decimal place calculation method, because in our CUDA settings, it does not support converting a floating-point value into a string.

The results, presented in Fig.~\ref{fig:6_3_2Ablation}(b), demonstrate that both of our core optimizations are critical for performance. First, {\em Falcon} achieves substantially better compression ratio and throughput than {\em Fal.}\textsubscript{Elf}. This demonstrates the effectiveness of our accurately decimal place calculations, leading to better compression. Second, our adaptive bit-plane encoding proves superior to both static strategies of {\em Fal.}\textsubscript{Sparse} and {\em Fal.}\textsubscript{Dense}. Although the adaptive approach incurs additional overhead, its ability to achieve a better compression ratio drastically reduces data transmission, leading to an overall throughput improvement compared to both static strategies.



\subsection{Performance for Single Values}
{\em Falcon} can be easily extended to single-precision floating-point data. Table~\ref{tab:floating} presents the results of our experiments on datasets with $\beta < 8$. We can see that {\em Falcon} achieves superior compression ratio and throughput. Specifically, {\em Falcon} outperforms its best competitor, Elf*, by an average relative compression ratio improvement of 31\%. Furthermore, {\em Falcon}'s compression throughput is significantly higher than the best-performing alternative, GDeflate, achieving an average speedup of $2.7\times$. The improved decompression throughput further demonstrates {\em Falcon}'s advantage, with an average increase of $2.7\times$ compared to GDeflate. These results highlight {\em Falcon}'s effectiveness in handling single-precision data.

\begin{table}[t!]
  \centering
  \captionsetup{skip=5pt} 
  \caption{Performance for Single Values.}
\label{tab:floating}
  \setlength{\tabcolsep}{2pt} 
  \resizebox{\columnwidth}{!}{  
  \begin{tabular}{@{}c|cccccccc|ccc@{}}
\bottomrule
\multirow{2}{*}{\textbf{\makecell[c]{Metrics}}} & \multicolumn{8}{c|}{\textbf{GPU}} & \multicolumn{3}{c}{\textbf{CPU}}                                         \\ \cline{2-12} 
       & {\em Falcon}  & ndzip & Lz4  & Bitc. & GDef. & Snap. & Elf* & ALP  & Elf  & Elf*      & ALP\\ \hline
\textbf{CR}                       & \textbf{0.279}      & 0.982      & 0.740      & 0.984 & 0.713  & 0.620 & \underline{0.404}    & 0.484     & 0.461   & 0.413 & 0.454 \\
\textbf{CT(GB/s)}                       &\textbf{11.86}       & 1.71      & 1.34      & 2.43 & \underline{4.38}  & 0.07 & 0.45    & 3.61     & 0.06   & 0.03 & 0.23 \\
\textbf{DT(GB/s)}                       &\textbf{11.57}       & 1.40      & 2.11      & 1.72 & \underline{4.21}  & 0.38 & 4.17    & 1.44     & 0.09   & 0.07 & 1.63 \\
\toprule
\end{tabular}
}
\end{table}

\section{related works}\label{sec:relatedwork}

\subsection{CPU-Based General Lossless Compression}
Traditional general-purpose lossless compression methods can also be applied to compress floating-point data. Common algorithms include Lz77~\cite{LZ77} and Huffman coding, which reduce data redundancy through sliding windows and optimal encoding methods, respectively. Deflate~\cite{deutsch1996deflate} is a widely used compression method based on Lz77 and Huffman coding, offering a moderate balance between compression efficiency and speed. Bzip2~\cite{seward1996bzip2}, on the other hand, uses Burrows-Wheeler Transform (BWT) and Run-Length Encoding (RLE), achieving higher compression ratios than Gzip but at a slower speed, making it more suitable for scenarios that require higher compression ratios.

Zstandard (ZSTD)~\cite{collet2016zstd} is a newer general-purpose compression algorithm that combines Lz77 with entropy coding, providing both high compression ratios and speed, and supports adjustable compression levels. Snappy~\cite{snappy} adopts a simple sliding window algorithm focused on compression speed and low latency, making it suitable for high-performance computing and real-time applications, though its compression ratio is lower than that of ZSTD. Additionally, Lz4~\cite{LZ4} is another fast compression algorithm focused on speed, especially suitable for real-time applications that require low latency and high throughput. 

Although these general algorithms perform well in many scenarios, they are not optimized for floating-point data, often resulting in limited compression ratios. This has led to the development of specialized compression algorithms for floating-point data.

\subsection{CPU-Based Floating-Point Lossless Compression}
Several specialized compression methods have been proposed in response to the characteristics of floating-point data, such as continuity and local correlation. The classic FPzip~\cite{lindstrom2006fast} algorithm leverages the prediction of residuals in floating-point data, followed by Huffman encoding for compression. ndzip~\cite{knorr2021ndzip} similarly uses differential calculation followed by bitmap encoding. Both PDE~\cite{kuschewski2023btrblocks} and ALP~\cite{afroozeh2023alp} employ the concept of Pseudo-Decimals, converting floating-point data into integers for further encoding and compression. The algorithms also encounter challenges due to computational errors inherent in floating-point data.

Gorilla~\cite{pelkonen2015gorilla} compresses adjacent floating-point values by XORing them with the previous value. Chimp and Chimp$_{128}$~\cite{liakos2022chimp} optimize Gorilla. Chimp$_{128}$ selects the best-matching historical value from the previous 128 values for XOR, increasing the number of trailing zeros in the XOR result. Elf~\cite{li2023elf} and Elf*~\cite{li2025adaptive} introduce an “erasing” strategy that removes selected bits to further increase trailing zeros, with Elf* combining adaptive Huffman coding and dynamic programming for improved compression. Its streaming version, SElf*~\cite{li2025adaptive}, further enhances efficiency in real-time processing and edge computing scenarios. Despite their effectiveness, the complexity of these encoding schemes limits parallelism.


\subsection{GPU-Based Floating-Point Lossless Compression}
With the advancement of GPU hardware, an increasing number of studies have focused on parallelizing the floating-point compression process on GPU platforms to overcome performance bottlenecks in large-scale data compression~\cite{hishida2025beyond}.


nvCOMP~\cite{nvcomp2023} is a compression library developed by NVIDIA that provides a collection of GPU-based compression algorithms.  It mainly adapts conventional and widely used compression schemes, such as LZ4, Deflate and Snappy, to run on GPUs. nvCOMP's advantages lie in its high throughput and parallelism, fully leveraging the CUDA architecture. However, there are still issues with compression ratio and throughput for general-purpose compression. In addition, nvCOMP also includes Bitcomp, a proprietary compressor designed for floating-point data.

GPU-based Floating-Point Compressor (GFC)~\cite{o2011floating} is a GPU-based algorithm that compresses double-precision data using residuals, but its fixed-size blocks limit input data to 512 MB. Massive Parallel Compression (MPC)~\cite{yang2015mpc} combines delta compression and data transformation. The method relies on large block sizes, which may lead to performance degradation in highly fluctuating data. ndzip-GPU~\cite{9910073} is the GPU version of the ndzip algorithm. However, ndzip is designed for high-dimensional floating-point data and is less effective for compressing low-dimensional data. Typed Data Transformation (TDT)~\cite{jamalidinan2025floating} is a data pre-processing method that rearranges bytes within floating-point numbers to enhance the efficiency of subsequent general-purpose compressors.

In summary, GPU-based floating-point data compression algorithms significantly improve compression and decompression efficiency through parallelized processing. However, they also face challenges such as reduced compression ratios and low throughput.

\section{Conclusion and Future Work}\label{sec:conclusion}

In this paper, we introduce {\em{Falcon}}, a novel GPU-based adaptive lossless floating-point compression framework.  Extensive experiments on 12 different datasets demonstrate that {\em Falcon} achieves an average compression ratio of 0.299, with average compression and decompression throughputs of 10.82 GB/s and 12.32 GB/s, respectively. Compared to the best-performing competitor, {\em Falcon} provides a 9.1\% relative improvement in compression ratio. Furthermore, it delivers a $2.4\times$ to $39\times$ speedup in compression throughput and a $2.7\times$ to $15\times$ speedup in decompression throughput over competing GPU-based solutions. In our future work, we will explore direct analytics on {\em Falcon}'s compressed data without full decompression.

\bibliographystyle{ACM-Reference-Format}
\bibliography{sample}

@String{Computing = "Computing" }

@String{Computer = "{IEEE} Computer" }

@article{hishida2025beyond,
  title={Beyond Compression: A Comprehensive Evaluation of Lossless Floating-Point Compression},
  author={Hishida, Kaisei and Liu, Chunwei and Paparrizos, John and Elmore, Aaron J},
  journal={Proceedings of the VLDB Endowment},
  volume={18},
  number={11},
  pages={4396--4409},
  year={2025},
  publisher={VLDB Endowment}
}

@article{yang2023most,
  title={MOST: model-based compression with outlier storage for time series data},
  author={Yang, Zehai and Chen, Shimin},
  journal={Proceedings of the ACM on Management of Data},
  volume={1},
  number={4},
  pages={1--29},
  year={2023},
  publisher={ACM New York, NY, USA}
}

@article{tang2025improving,
  title={Improving Time Series Data Compression in Apache IoTDB},
  author={Tang, Yuxin and Zhang, Feng and Guan, Jiawei and Tian, Yuan and Huang, Xiangdong and Wang, Chen and Wang, Jianmin and Du, Xiaoyong},
  journal={Proceedings of the VLDB Endowment},
  volume={18},
  number={10},
  pages={3406--3420},
  year={2025},
  publisher={VLDB Endowment}
}

@techreport{deutsch1996deflate,
  title={DEFLATE compressed data format specification version 1.3},
  author={Deutsch, Peter},
  year={1996}
}

@misc{snappy,
  title={Snappy | A fast compressor/decompressor},
  note={Retrieved July 26, 2023 from \url{https://github.com/google/snappy}},
  author={Google},
  year={2023}
}

@article{seward1996bzip2,
  title={bzip2 and libbzip2},
  author={Seward, Julian},
  journal={avaliable at http://www. bzip. org},
  pages={8--18},
  year={1996}
}

@misc{collet2016zstd,
  title={Zstd github repository from facebook},
  note={Retrieved July 26, 2023 from \url{https://github.com/facebook/zstd}},
  author={Collet, Y},
  year={2016}
}

@article{LZ77,
  author={Ziv, J. and Lempel, A.},
  journal={IEEE Transactions on Information Theory}, 
  title={A universal algorithm for sequential data compression}, 
  year={1977},
  volume={23},
  number={3},
  pages={337-343},
  doi={10.1109/TIT.1977.1055714}
}

@misc{LZ4,
  author       = {Yann Collet},
  title        = {Lz4: Extremely fast compression algorithm},
  year         = {2013},
  howpublished = {\url{https://github.com/lz4/lz4}},
  note         = {Accessed: April 28, 2025}
}

@article{lindstrom2006fast,
  title={Fast and efficient compression of floating-point data},
  author={Lindstrom, Peter and Isenburg, Martin},
  journal={IEEE transactions on visualization and computer graphics},
  volume={12},
  number={5},
  pages={1245--1250},
  year={2006},
  publisher={IEEE}
}

@INPROCEEDINGS{9910073,
  author={Knorr, Fabian and Thoman, Peter and Fahringer, Thomas},
  booktitle={SC21: International Conference for High Performance Computing, Networking, Storage and Analysis}, 
  title={ndzip-gpu: Efficient lossless Compression of Scientific Floating-Point Data on GPUs}, 
  year={2021},
  volume={},
  number={},
  pages={1-13},
  doi={10.1145/3458817.3476224}}

@inproceedings{knorr2021ndzip,
  title={ndzip: A high-throughput parallel lossless compressor for scientific data},
  author={Knorr, Fabian and Thoman, Peter and Fahringer, Thomas},
  booktitle={2021 Data Compression Conference (DCC)},
  pages={103--112},
  year={2021},
  organization={IEEE}
}

@article{kuschewski2023btrblocks,
  title={Btrblocks: Efficient columnar compression for data lakes},
  author={Kuschewski, Maximilian and Sauerwein, David and Alhomssi, Adnan and Leis, Viktor},
  journal={Proceedings of the ACM on Management of Data},
  volume={1},
  number={2},
  pages={1--26},
  year={2023},
  publisher={ACM New York, NY, USA}
}

@article{afroozeh2023alp,
  title={Alp: Adaptive lossless floating-point compression},
  author={Afroozeh, Azim and Kuffo, Leonardo X and Boncz, Peter},
  journal={Proceedings of the ACM on Management of Data},
  volume={1},
  number={4},
  pages={1--26},
  year={2023},
  publisher={ACM New York, NY, USA}
}

@article{pelkonen2015gorilla,
  title={Gorilla: A fast, scalable, in-memory time series database},
  author={Pelkonen, Tuomas and Franklin, Scott and Teller, Justin and Cavallaro, Paul and Huang, Qi and Meza, Justin and Veeraraghavan, Kaushik},
  journal={Proceedings of the VLDB Endowment},
  year={2015},
  publisher={VLDB Endowment}
}

@article{liakos2022chimp,
  title={Chimp: efficient lossless floating point compression for time series databases},
  author={Liakos, Panagiotis and Papakonstantinopoulou, Katia and Kotidis, Yannis},
  journal={Proceedings of the VLDB Endowment},
  volume={15},
  number={11},
  pages={3058--3070},
  year={2022},
  publisher={VLDB Endowment}
}

@article{li2023elf,
  title={Elf: Erasing-Based Lossless Floating-Point Compression},
  author={Li, Ruiyuan and Li, Zheng and Wu, Yi and Chen, Chao and Zheng, Yu},
  journal={Proceedings of the VLDB Endowment},
  volume={16},
  number={7},
  pages={1763--1776},
  year={2023},
  publisher={VLDB Endowment}
}

@article{Elf+,
  title={Erasing-based lossless compression method for streaming floating-point time series},
  author={Li, Ruiyuan and Li, Zheng and Wu, Yi and Chen, Chao and Guo, Songtao and Zhang, Ming and Zheng, Yu},
  journal={arXiv preprint arXiv:2306.16053},
  year={2023}
}

@inproceedings{o2011floating,
  title={Floating-point data compression at 75 Gb/s on a GPU},
  author={O'Neil, Molly A and Burtscher, Martin},
  booktitle={Proceedings of the Fourth Workshop on General Purpose Processing on Graphics Processing Units},
  pages={1--7},
  year={2011}
}

@inproceedings{yang2015mpc,
  title={MPC: a massively parallel compression algorithm for scientific data},
  author={Yang, Annie and Mukka, Hari and Hesaaraki, Farbod and Burtscher, Martin},
  booktitle={2015 IEEE International Conference on Cluster Computing},
  pages={381--389},
  year={2015},
  organization={IEEE}
}

@misc{nvcomp2023,
  author = {{NVIDIA Corporation}},
  title = {nvCOMP},
  year = {2023},
  howpublished = {\url{https://github.com/NVIDIA/nvcomp}},
  note = {Accessed: April 28, 2025}
}

@misc{gzip,
  title        = {Gzip: The GNU data compression utility},
  howpublished = {\url{https://www.gnu.org/software/gzip/}},
  note         = {Accessed: May 31, 2025},
  year         = {n.d.}
}

@inproceedings{huang2024cuszp2,
  author    = {Huang, Yafan and Di, Sheng and Li, Guanpeng and Cappello, Franck},
  title     = {{CUSZP2: A GPU Lossy Compressor with Extreme Throughput and Optimized Compression Ratio}},
  booktitle = {Proceedings of the International Conference for High Performance Computing, Networking, Storage and Analysis (SC '24)},
  year      = {2024},
  month     = {November},
  address   = {Atlanta, GA, USA},
  publisher = {IEEE / ACM},
}

@techreport{merrill2016decoupled,
  author      = {Merrill, Duane and Garland, Michael},
  title       = {Single-pass Parallel Prefix Scan with Decoupled Look-back},
  institution = {{NVIDIA Corporation}},
  year        = {2016},
  number      = {NVR-2016-002},
  month       = {March},
  type        = {Technical Report},
}

@misc{elf_stocks_2023,
  author = {{INFORE project}},
  title  = {Financial data set used in {INFORE} project},
  year   = {2023},
  url    = {https://zenodo.org/record/3886895},
  note   = {Accessed: March 19, 2023}
}

@misc{serf_tsbs_2023,
    author = {{Timescale}},
    title  = {Time Series Benchmark Suite ({TSBS})},
    year   = {2023},
    url    = {https://github.com/timescale/tsbs}
}

@misc{neon_pressure_2022,
  author       = {{National Ecological Observatory Network (NEON)}},
  title        = {Barometric pressure ({DP1.00004.001})},
  year         = {2022},
  howpublished = {\url{https://data.neonscience.org/data-products/DP1.00004.001}},
  note         = {Dataset. DOI: 10.48443/ZR37-0238. Accessed on: April 30, 2024}
}

@misc{neon_windspeed_2022,
  author       = {{National Ecological Observatory Network (NEON)}},
  title        = {2D wind speed and direction ({DP1.00001.001})},
  year         = {2022},
  howpublished = {\url{https://data.neonscience.org/data-products/DP1.00001.001}},
  note         = {Dataset. DOI: 10.48443/77N6-EH42. Accessed on: April 30, 2024}
}

@article{li2025serf,
  title={Serf: Streaming Error-Bounded Floating-Point Compression},
  author={Li, Ruiyuan and Chen, Zechao and Lu, Ruyun and Xu, Xiaolong and Yang, Guangchao and Chen, Chao and Bao, Jie and Zheng, Yu},
  journal={Proceedings of the ACM on Management of Data},
  volume={3},
  number={3},
  pages={1--27},
  year={2025},
  publisher={ACM New York, NY, USA}
}

@article{jamalidinan2025floating,
  title={Floating-Point Data Transformation for Lossless Compression},
  author={Jamalidinan, Samirasadat and Cheshmi, Kazem},
  journal={arXiv preprint arXiv:2506.18062},
  year={2025}
}

@article{li2025adaptive,
  title={Adaptive Encoding Strategies for Lossless Floating-Point Compression},
  author={Li, Zheng and Li, Ruiyuan and Xu, Xiaolong and Wu, Yi and Chen, Chao and Liu, Tong and Shang, Jiaxing and Zheng, Yu},
  journal={IEEE Internet of Things Journal},
  year={2025},
  publisher={IEEE}
}

@article{jensen2017time,
  title={Time series management systems: A survey},
  author={Jensen, S{\o}ren Kejser and Pedersen, Torben Bach and Thomsen, Christian},
  journal={IEEE Transactions on Knowledge and Data Engineering},
  volume={29},
  number={11},
  pages={2581--2600},
  year={2017},
  publisher={IEEE}
}

@misc{Zigzag2001, 
	note = {\url{https://protobuf.dev/programming-guides/encoding/}}, 
	year = {2001}, 
	title = {Protocol buffers encoding}, 
	author = {Google} 
}

@article{chen2024fcbench,
  title={FCBench: Cross-Domain Benchmarking of Lossless Compression for Floating-Point Data},
  author={Chen, Xinyu and Tian, Jiannan and Beaver, Ian and Freeman, Cynthia and Yan, Yan and Wang, Jianguo and Tao, Dingwen},
  journal={Proceedings of the VLDB Endowment},
  volume={17},
  number={6},
  pages={1418--1431},
  year={2024},
  publisher={VLDB Endowment}
}

@article{liu2021high,
  title={High-ratio lossy compression: Exploring the autoencoder to compress scientific data},
  author={Liu, Tong and Wang, Jinzhen and Liu, Qing and Alibhai, Shakeel and Lu, Tao and He, Xubin},
  journal={IEEE Transactions on Big Data},
  volume={9},
  number={1},
  pages={22--36},
  year={2021},
  publisher={IEEE}
}

@inproceedings{di2016fast,
	title={Fast error-bounded lossy HPC data compression with SZ},
	author={Di, Sheng and Cappello, Franck},
	booktitle={2016 ieee international parallel and distributed processing symposium (ipdps)},
	pages={730--739},
	year={2016},
	organization={IEEE}
}

@inproceedings{liang2018error,
	title={Error-controlled lossy compression optimized for high compression ratios of scientific datasets},
	author={Liang, Xin and Di, Sheng and Tao, Dingwen and Li, Sihuan and Li, Shaomeng and Guo, Hanqi and Chen, Zizhong and Cappello, Franck},
	booktitle={2018 IEEE International Conference on Big Data (Big Data)},
	pages={438--447},
	year={2018},
	organization={IEEE}
}

@inproceedings{zhao2020significantly,
  title={Significantly improving lossy compression for HPC datasets with second-order prediction and parameter optimization},
  author={Zhao, Kai and Di, Sheng and Liang, Xin and Li, Sihuan and Tao, Dingwen and Chen, Zizhong and Cappello, Franck},
  booktitle={Proceedings of the 29th International Symposium on High-Performance Parallel and Distributed Computing},
  pages={89--100},
  year={2020}
}

@article{liang2022sz3,
  title={SZ3: A modular framework for composing prediction-based error-bounded lossy compressors},
  author={Liang, Xin and Zhao, Kai and Di, Sheng and Li, Sihuan and Underwood, Robert and Gok, Ali M and Tian, Jiannan and Deng, Junjing and Calhoun, Jon C and Tao, Dingwen and others},
  journal={IEEE Transactions on Big Data},
  volume={9},
  number={2},
  pages={485--498},
  year={2022},
  publisher={IEEE}
}

@inproceedings{shi2024machete,
	title={Machete: An Efficient Lossy Floating-Point Compressor Designed for Time Series Databases},
	author={Shi, Yang and Zou, Xiangyu and Chen, Xinyu and Jin, Sian and Tao, Dingwen and Cai, Deng and Chen, Yufan and Xia, Wen},
	booktitle={2024 Data Compression Conference (DCC)},
	year={2024}
}

@article{wang2020apache,
  title={Apache IoTDB: Time-series database for internet of things},
  author={Wang, Chen and Huang, Xiangdong and Qiao, Jialin and Jiang, Tian and Rui, Lei and Zhang, Jinrui and Kang, Rong and Feinauer, Julian and McGrail, Kevin A and Wang, Peng and others},
  journal={Proceedings of the VLDB Endowment},
  volume={13},
  number={12},
  pages={2901--2904},
  year={2020},
  publisher={VLDB Endowment}
}

@inproceedings{li2020just,
  title={Just: Jd urban spatio-temporal data engine},
  author={Li, Ruiyuan and He, Huajun and Wang, Rubin and Huang, Yuchuan and Liu, Junwen and Ruan, Sijie and He, Tianfu and Bao, Jie and Zheng, Yu},
  booktitle={2020 IEEE 36th International Conference on Data Engineering (ICDE)},
  pages={1558--1569},
  year={2020},
  organization={IEEE}
}

@article{yao2024camel,
  title={Camel: Efficient Compression of Floating-Point Time Series},
  author={Yao, Yuanyuan and Chen, Lu and Fang, Ziquan and Gao, Yunjun and Jensen, Christian S and Li, Tianyi},
  journal={Proceedings of the ACM on Management of Data},
  volume={2},
  number={6},
  pages={1--26},
  year={2024},
  publisher={ACM New York, NY, USA}
}

@article{kahan1996ieee,
	title={IEEE standard 754 for binary floating-point arithmetic},
	author={Kahan, William},
	journal={Lecture Notes on the Status of IEEE},
	volume={754},
	number={94720-1776},
	pages={11},
	year={1996}
}

@misc{Kaggle_city_temp2019,
  author       = {{Kaggle}},
  title        = {Climate Weather Surface of Brazil - Hourly},
  year         = {2019},
  howpublished = {\url{https://www.kaggle.com/datasets/PROPPG-PPG/hourly-weather-surface-brazil-southeast-region}},
  note         = {Accessed February 13, 2025}
}

@misc{gas_sensor_2015,
  author       = {Fonollosa, Jordi},
  title        = {{Gas sensor array under dynamic gas mixtures}},
  year         = {2015},
  howpublished = {UCI Machine Learning Repository},
  note         = {{DOI}: https://doi.org/10.24432/C5WP4C}
}

@misc{JaneStreetKaggle,
  author       = {{Mohamed Sameh}},
  title        = {Jane Street Dataset},
  year         = {n.d.},
  howpublished = {\url{https://www.kaggle.com/datasets/mohamedsameh0410/jane-street-dataset}},
  note         = {Accessed April 28, 2025}
}

@misc{solar_wind_Kaggle,
  author       = {{kingabzpro}},
  title        = {MAGNET NASA Dataset},
  year         = {n.d.},
  howpublished = {\url{https://www.kaggle.com/datasets/kingabzpro/magnet-nasa?select=solar_wind.csv}},
  note         = {Accessed April 28, 2025}
}

@article{NYX_Almgren_2013,
doi = {10.1088/0004-637X/765/1/39},
url = {https://dx.doi.org/10.1088/0004-637X/765/1/39},
year = {2013},
month = {feb},
publisher = {The American Astronomical Society},
volume = {765},
number = {1},
pages = {39},
author = {Almgren, Ann S. and Bell, John B. and Lijewski, Mike J. and Lukić, Zarija and Van Andel, Ethan},
title = {Nyx: A MASSIVELY PARALLEL AMR CODE FOR COMPUTATIONAL COSMOLOGY},
journal = {The Astrophysical Journal},
}

@misc{NYCTaxi2021,
  author       = {{Kaggle}},
  title        = {NYC Yellow Taxi Trip Data},
  year         = {2021},
  howpublished = {\url{https://www.kaggle.com/datasets/elemento/nyc-yellow-taxi-trip-data}},
  note         = {Accessed February 13, 2024}
}

\end{document}